\documentclass[acmsmall]{acmart}

\usepackage{booktabs}
\usepackage{graphicx}
\usepackage{multirow}
\usepackage{algorithm}
\usepackage{enumitem}
\usepackage{algorithmic}
\usepackage{balance}
\usepackage{arydshln}
\usepackage{stfloats}
\usepackage{bm}
\usepackage{subfigure}
\usepackage{varwidth}
\usepackage[normalem]{ulem}
\useunder{\uline}{\ul}{}

\AtBeginDocument{%
  \providecommand\BibTeX{{%
    \normalfont B\kern-0.5em{\scshape i\kern-0.25em b}\kern-0.8em\TeX}}}

\setcopyright{none}

\acmYear{2023}
\acmDOI{XXXXXXX.XXXXXXX}

\acmJournal{JACM}
\acmVolume{37}
\acmNumber{4}
\acmArticle{1}
\acmMonth{8}

\begin{document}

\title{DPR: An Algorithm Mitigate Bias Accumulation in Recommendation feedback loops}
\author{Hangtong Xu}
\affiliation{%
  \institution{MIC Lab, College of Computer Science and Technology, Jilin University}
  \city{Changchun}
  \country{China}
  }
\email{xuht21@mails.jlu.edu.cn}
\author{Yuanbo Xu}
\affiliation{%
  \institution{MIC Lab, College of Computer Science and Technology, Jilin University}
  \city{Changchun}
  \country{China}
  }
\email{yuanbox@jlu.edu.cn}
\authornote{Corresponding Author: Yuanbo Xu}
\author{Yongjian Yang}
\affiliation{%
  \institution{MIC Lab, College of Computer Science and Technology, Jilin University}
  \city{Changchun}
  \country{China}
  }
\email{yyj@jlu.edu.cn}
\author{Fuzhen Zhuang}
\affiliation{%
  \institution{Institute of Artificial Intelligence, Beihang University}
  \city{Beijing}
  \country{China}
  }
\email{zhuangfuzhen@buaa.edu.cn}
\author{Hui Xiong}
\affiliation{%
  \institution{Hong Kong University of Science and Technology (Guangzhou)}
  \city{Guangzhou}
  \country{China}
  }
\email{xionghui@ust.hk}
\renewcommand{\shortauthors}{Hangtong and Yuanbo, et al.}
\begin{abstract}
  Recommendation models trained on the user feedback collected from deployed recommendation systems are commonly biased. User feedback is considerably affected by the exposure mechanism, as users only provide feedback on the items exposed to them and passively ignore the unexposed items, thus producing numerous false negative samples. Inevitably, biases caused by such user feedback are inherited by new models and amplified via feedback loops. Moreover, the presence of false negative samples makes negative sampling difficult and introduces spurious information in the user preference modeling process of the model. Recent work has investigated the negative impact of feedback loops and unknown exposure mechanisms on recommendation quality and user experience, essentially treating them as independent factors and ignoring their cross-effects. To address these issues, we deeply analyze the data exposure mechanism from the perspective of data iteration and feedback loops with the Missing Not At Random (\textbf{MNAR}) assumption, theoretically demonstrating the existence of an available stabilization factor in the transformation of the exposure mechanism under the feedback loops. We further propose Dynamic Personalized Ranking (\textbf{DPR}), an unbiased algorithm that uses dynamic re-weighting to mitigate the cross-effects of exposure mechanisms and feedback loops without additional information. Furthermore, we design a plugin named Universal Anti-False Negative (\textbf{UFN}) to mitigate the negative impact of the false negative problem. We demonstrate theoretically that our approach mitigates the negative effects of feedback loops and unknown exposure mechanisms. Experimental results on real-world datasets demonstrate that models using DPR can better handle bias accumulation and the universality of UFN in mainstream loss methods. 
\end{abstract}




\maketitle
\section{Introduction}
Recommendation systems facilitate people's lives \cite{tois8,tois4}. They change the way of retrieving information through precise personalization and enhance the user experience. Existing models are dedicated to extracting user preferences from data collected from user feedback on previous system recommendations \cite{tois3,tois10,tois7,tois9} and generating new recommendations in the next stage. The interaction processes contained there induce a feedback loop: the recommendation model affects the user behavior data it observes, and the data affects the model trained on these data. As the model iterates, increasing bias accumulates in the data, hiding the performance of the trained model. \cite{tois6,tois2}.\par
\begin{figure}[tp]
\centering
\begin{minipage}[t]{0.48\linewidth}
\centering
\centerline{\includegraphics[height = 5cm]{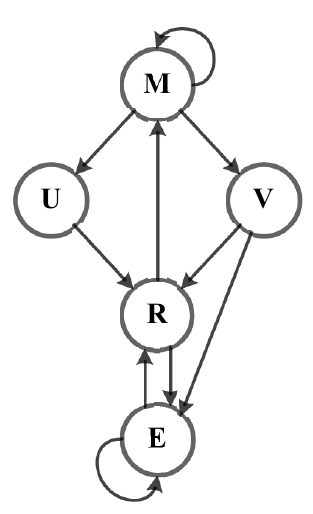}}
\centerline{(a)}
\end{minipage}
\hfill
\begin{minipage}[t]{0.48\linewidth}
\centering
\centerline{\includegraphics[height = 5cm]{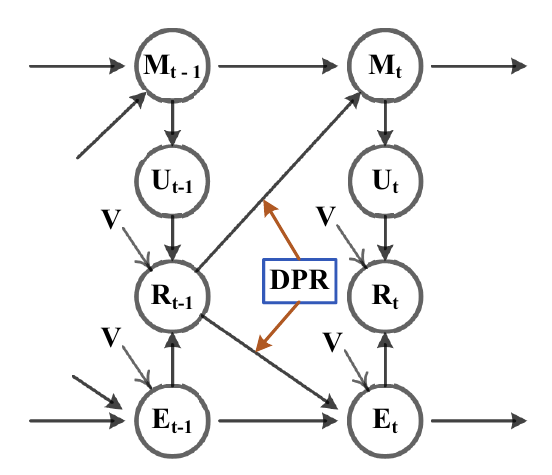}}
\centerline{(b)}
\centering
\end{minipage}
\caption{\textbf{U, V} are the intrinsic properties of users and items modeled by Model \textbf{M}. \textbf{R} is the user-item relevance score, and \textbf{E} is the exposure mechanism. (a) is a graph to show the cross-effects of the feedback loops and unknown exposure mechanisms on the modeling process. (b) is the time-expanded graph of (a).}
\label{fig:1}
\end{figure}
Furthermore, recent work \cite{tois5,ref4,ref5,ref12} shows that most systems rely on the assumption that user feedback is Missing At Random (\textbf{MAR}), which means that items are equally possible to be observed by users. However, in most scenarios, users cannot be exposed to all items, especially for e-commerce, there is a significant imbalance in the ratio of users to items. Thus, real-world data generally obeys the Missing Not At Random (\textbf{MNAR}) assumption. This means that missing interactions in the user feedback data cannot simply be considered negative, since items may be preferred by users but are not exposed to them. In this case, previous researchers directly treat unobserved items as negative samples, creating the false-negative sample problem and preventing the model from properly modeling user preferences. In recommendation systems, the MNAR assumption can be translated into the exposure mechanism, which determines the possibility of the item being exposed to users.\par
The impact of the feedback loops and the exposure mechanism on the recommendation system is not independent. The model (\textbf{M}) calculates the relevance score (\textbf{R}) of user-item pairs for ranking and recommends high-rankers to the user, thus changing the initial exposure mechanism (\textbf{E}) that items have the same possibility to be exposed to the user in the user-item interaction phase, the presence of the feedback loops makes the bias accumulate and get serious. A typical graph is displayed in Figure \ref{fig:1}. For example, early popular items are more likely to be recommended to users than cold start items, because the recommendation system passively observes more interaction data for popular items, which causes the “rich-get-richer” problem.\par
In this work, we study the exposure mechanism under the feedback loops with the MNAR assumption and propose Dynamic Personal Ranking (\textbf{DPR}), an algorithm that can provably mitigate the cross-effects of feedback loops and exposure mechanisms in recommendation systems. We have a key observation: the recommendation system does have an impact on the exposure mechanism and keeps accumulating in the feedback loops, and an available stabilization factor in the transformation of the exposure mechanism under the feedback loops. Bias in feedback data collected from recommendation systems increases and is inherited by the new model as the length of the feedback loops increase. Breaking the chain of influence of feedback data collected in the feedback loops on the model and the exposure mechanism is a possible solution to free recommendations from bias accumulation. To achieve this goal, we design an unbiased algorithm, DPR, a pairwise loss function that uses dynamic re-weighting to eliminate accumulation bias in the data, as shown in Figure \ref{fig:1} (b). DPR employs the stabilization factor in the above-mentioned exposure mechanism transition process as the weight, enabling the model to be free from accumulated bias in data and modeling true user preference without additional information, no matter how long the feedback loops. In addition, we introduce the Universal Anti-False Negative (\textbf{UFN}), a removable plug-in to solve the common problem of false negatives in pairwise ranking methods. UFN enables the model to combat false negative samples and reduce the impact of the false information they carry.\par
The main contributions of our work are summarized as follows:\par
\begin{itemize}
    \item We use theoretical and empirical analysis to reveal the cross-effects of feedback loops and exposure mechanisms on recommendation systems. We find an available stabilization factor in the transformation of the exposure mechanism under the feedback loops.\par
    \item We propose a pairwise loss function named Dynamic Personalized Ranking (\textbf{DPR}) to eliminate accumulation bias in real-world datasets and achieves unbiased recommendation with theoretical accuracy guarantees.\par
    \item We introduce a plug-in named Universal Anti-False Negative (\textbf{UFN}) to mitigate the false negative sample problem to improve training quality, help the model combat false negative samples and reduce the impact of the false information they carry.\par
    \item We conduct empirical studies on the simulation dataset and six real-world datasets to demonstrate DPR can mitigate bias accumulation in recommendation feedback loops and the universality of UFN in mainstream loss methods.\par
    \item We conduct experiments on widely used models such as MLPs, and DIN to demonstrate the ease of use and effectiveness of DPR. 
\end{itemize}
\begin{table}[tp]
\renewcommand\arraystretch{1.5}
    \centering
    \caption{Comparison between different Pairwise Ranking methods.}
    \label{tab:1}
\resizebox{0.7\linewidth}{!}{
    \begin{tabular}[b]{cccccc}
    \bottomrule
    &\textbf{MANR Assumption} & \textbf{Exposure Mechanism} & \textbf{Negative Sample} & \textbf{Debias} & \textbf{Backbone} \\
    \midrule
    BPR&  & & & & MF\\
    \midrule
    UBPR& \footnotesize\checkmark  & \footnotesize\checkmark& &\footnotesize\checkmark & MF\\
    \midrule
    EBPR&  &  & &\footnotesize\checkmark & MF\\
    \midrule
    CPR&  &  &\footnotesize\checkmark &\footnotesize\checkmark & MF\\
    \midrule
    PDA&  & & & \footnotesize\checkmark & MF\\
    \midrule
    UPL& \footnotesize\checkmark & \footnotesize\checkmark&  & \footnotesize\checkmark & MF\\
    \midrule
    DPR& \footnotesize\checkmark & \footnotesize\checkmark& \footnotesize\checkmark& \footnotesize\checkmark & MF\\
    \bottomrule
    \end{tabular}}
\end{table}
\section{Background}
In this section, we briefly review previous work from three perspectives: feedback loops, MNAR problems, and pairwise ranking.
\subsubsection*{Feedback loops} Recent research has shown that ignoring feedback effects will negatively impact recommendation system performance \cite{ref1,ref2,ref16,ref17}. Feedback loops can introduce bias in the collected user feedback data, such as popularity bias, selection bias, and homogenization \cite{ref18,ref16}. With the growth of loops, the bias is amplified and accumulates \cite{ref19}. Therefore, breaking the loop to correct feedback effects is a direct way to solve the problem. A common solution is to use the causal mechanism to break the loop \cite{ref13,ref14,ref25} so that the model does not suffer from feedback loops in muti-step recommendations or combine inverse propensity weighting (IPW) with active learning to correct for feedback effects \cite{ref20,ref21,xu01}. However, they focus on the feedback loops without considering the exposure mechanism, which is equally important.
\subsubsection*{MNAR problem} The missing–not-at-random (MNAR) problem has been extensively studied in recommendation systems: recommendation systems aim to infer missing ratings over a static dataset \cite{ref22,ref23,ref24}. One prevalent solution is Inverse Propensity Scoring (IPS) \cite{ref6}, where popular items are down-weighted and unpopular items are up-weighted, but if the model cannot accurately estimate the propensity score of each sample, it can not reach the theoretically unbiased. Simply using the inverse propensity score as a substitute for the exposure mechanism ignores the dynamic nature of the exposure mechanism under feedback loops. Hence, IPS-based methods do not properly deal with the effect of unknown exposure mechanisms.
\subsubsection*{Pairwise ranking} There are two types of loss functions that are generally used to optimize models: pointwise loss and pairwise loss. Pairwise loss received more praise due to its high accuracy compared to the pointwise loss. Pairwise loss learns user preferences over two items by maximizing the prediction of positive items over negative ones. Bayesian Personalized Ranking (BPR) \cite{ref8} is the first loss function that uses a pairwise ranking setting. Since it does not consider the bias in data, there are many improved versions of it, such as Unbiased Bayesian Personalized Ranking (UBPR) \cite{ref4}, Cross Pairwise Ranking (CPR) \cite{ref7} and Explainable Bayesian Personalized Ranking (EBPR) \cite{ref12}. Although they achieve unbiased recommendations, they ignore the feedback loops that affect the performance of the recommendation system and suffer from false negative problems. Here, we briefly summarized the common pairwise ranking methods in Table \ref{tab:1} with respect to four dimensions.
\section{Problem Formulation}
In this section, we first show how the feedback loops affect the exposure mechanism, and provide an in-depth analysis of the exposure mechanism, identify the key factors that cause bias. Then, we demonstrate existing pairwise ranking methods are biased in the presence of feedback loops. We give the notation in Table \ref{tab:2}.
\begin{table}[tp]
\renewcommand\arraystretch{1.25}
    \centering
    \caption{Notation.}
    \label{tab:2}
    \scalebox{1}{
    \begin{tabular}[b]{ll}
    \bottomrule
    \makebox[0.05\textwidth][l]{Symbol} & \makebox[0.3\textwidth][c]{Description} \\
    \toprule
    $M$& The number of users. \\
    $N$& The number of items. \\
    $\bm{S}$ & $M \times N$  rating matrix. \\
    $\bm{Y}$ & $M \times N$  recommendation matrix.\\
    $\bm{R}_{u,i}$ & The relevance between user u and item i.\\
    $\bm{O}(i)$ & The probability of the item i exposed to the users.\\
    $\bm{U}$& Low dimensional user factors.\\
    $\bm{V}$& Low dimensional item factors.\\
    $T$ & The feedback loops length, larger T, implies longer feedback loops.\\
    \bottomrule
    \end{tabular}}
\end{table}
\subsection{Feedback loops in Recommendation system}
Recommendation systems make recommendations and collect user feedback, and then train the next model based on collected feedback, a process that induces a feedback loop (Figure \ref{fig:1}). The recommendation system forces the exposure mechanism to change the probability of an item being exposed to the user. For example, a personalized recommendation system aims to recommend items that the user prefers as much as possible. Such items will have a higher probability of being exposed to users than others. As a result, we can observe more interactions of such items in the collected feedback, and the next model trained on such feedback might infer that the user likes it and continue to recommend it. However, the recommendation system cannot know all user preferences perfectly which leads to biased feedback data. The new model optimizes its parameter $\hat{\Theta}$ can be seen as a maximum likelihood estimator under the feedback data:
\begin{equation}
\begin{aligned} 
\hat{\Theta} &=\underset{\Theta}{\arg \min } \mathbb{E}_{\bm{Y}}\left\{\mathbb{K} \mathbb{L}\left[P\left(\bm{S} \mid \bm{Y}\right) \| P_{\Theta}\left(\bm{S} \mid \bm{Y}\right)\right]\right\} \\ &=\underset{\Theta}{\arg \max } \mathbb{E}_{\bm{Y}}\left\{\mathbb{E}_{P\left(\bm{S} \mid \bm{Y}\right)}\left[\log P_{\Theta}\left(\bm{S} \mid \bm{Y}\right)\right]\right\},
\end{aligned}
\label{eq:1}
\end{equation}
where {\small $P\left(\bm{S} \mid \bm{Y}\right)$} indicates the user preference we can learn from the feedback data, and {\small $P_{\Theta}\left(\bm{S} \mid \bm{Y}\right)$} is the inference of user preference by the new model. It is obvious from Equation \ref{eq:1} that models try to fit the feedback data, which means that models will inherit the biased user preferences in the collected feedback data and make sub-optimal recommendations. From this, we make the assumption that the original recommendation system does affect the exposure mechanism, which is reflected in the feedback data and will be inherited by the new model in the feedback loops.\par
To validate our assumption, we choose two datasets for analysis: Coat\footnote{\href{https://www.cs.cornell.edu/ schnabts/mnar/}{https://www.cs.cornell.edu/ schnabts/mnar/}} and KuaiRec\footnote{\href{https://chongminggao.github.io/KuaiRec/}{https://chongminggao.github.io/KuaiRec/}} \cite{kuai}. Both have full-observed and partial-observed data, which better reflect the difference in exposure mechanisms under and out of the feedback loops. Coat obtains fully observed data from questionnaires under randomized control and thus does not suffer from feedback loops, while KuaiRec does the opposite and obtains data from an online recommendation system without randomized control where impact of the recommendation system and feedback loops cannot be ignored. We choose the positive interaction counts of items in feedback data as the dominant indicator of the exposure mechanism, treating the full observed data as Missing At Random (MAR) data and the partial-observed data as Missing Not At Random (MNAR) data for Coat, and both the full observed data and the partial-observed data in Kuairec are Missing Not At Random (MNAR) data. The MNAR and MAR data are representative of the biased and unbiased exposure mechanisms, respectively. We process both datasets into implicit feedback without loss of generality, the results are shown in Figure \ref{fig:2}. \par
\begin{figure}[tp]
\centering
\begin{minipage}[t]{0.48\linewidth}
\centering
\includegraphics[width= 0.9\linewidth]{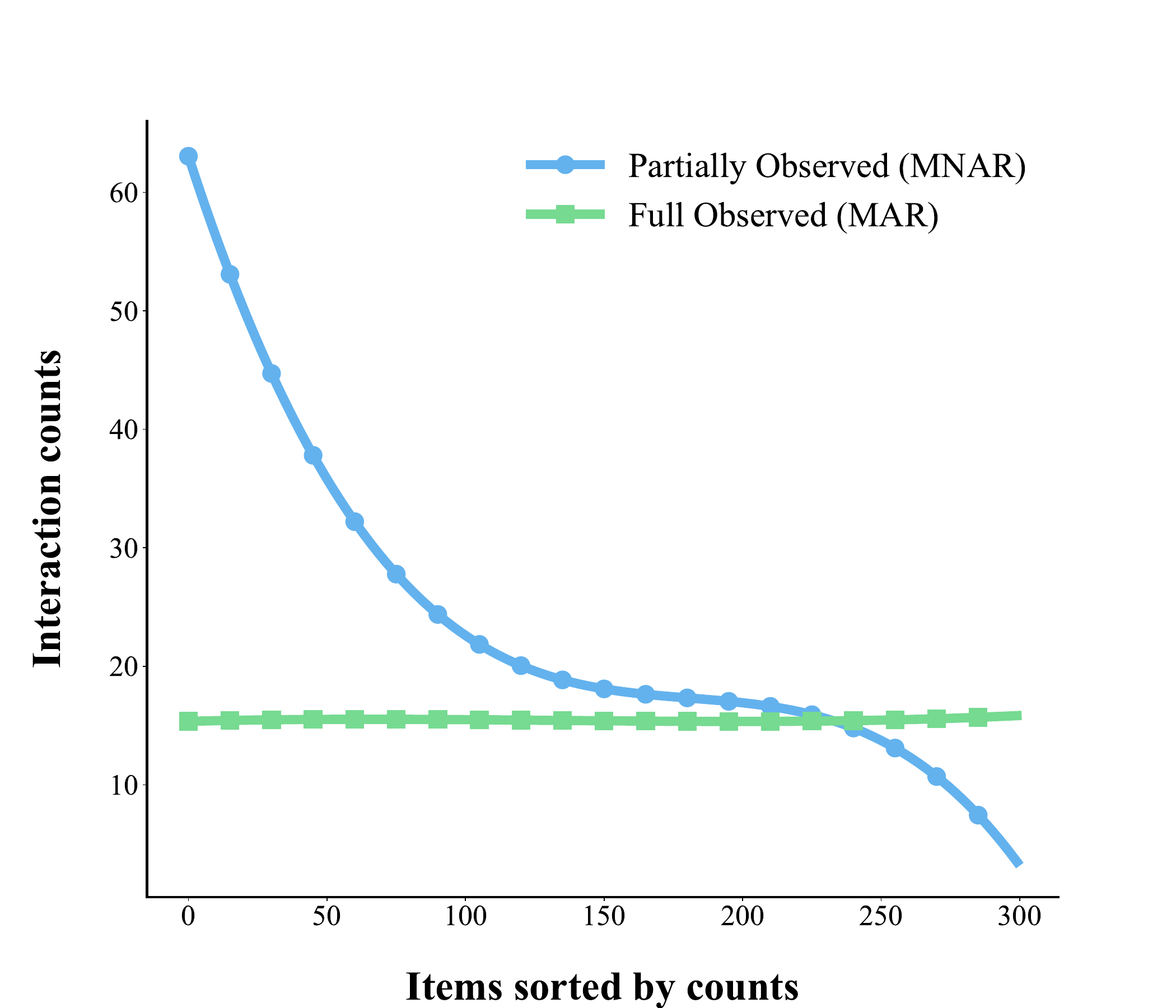}
\centerline{(a) Coat}
\end{minipage}
\begin{minipage}[t]{0.48\linewidth}
\centering
\includegraphics[width= 0.9\linewidth]{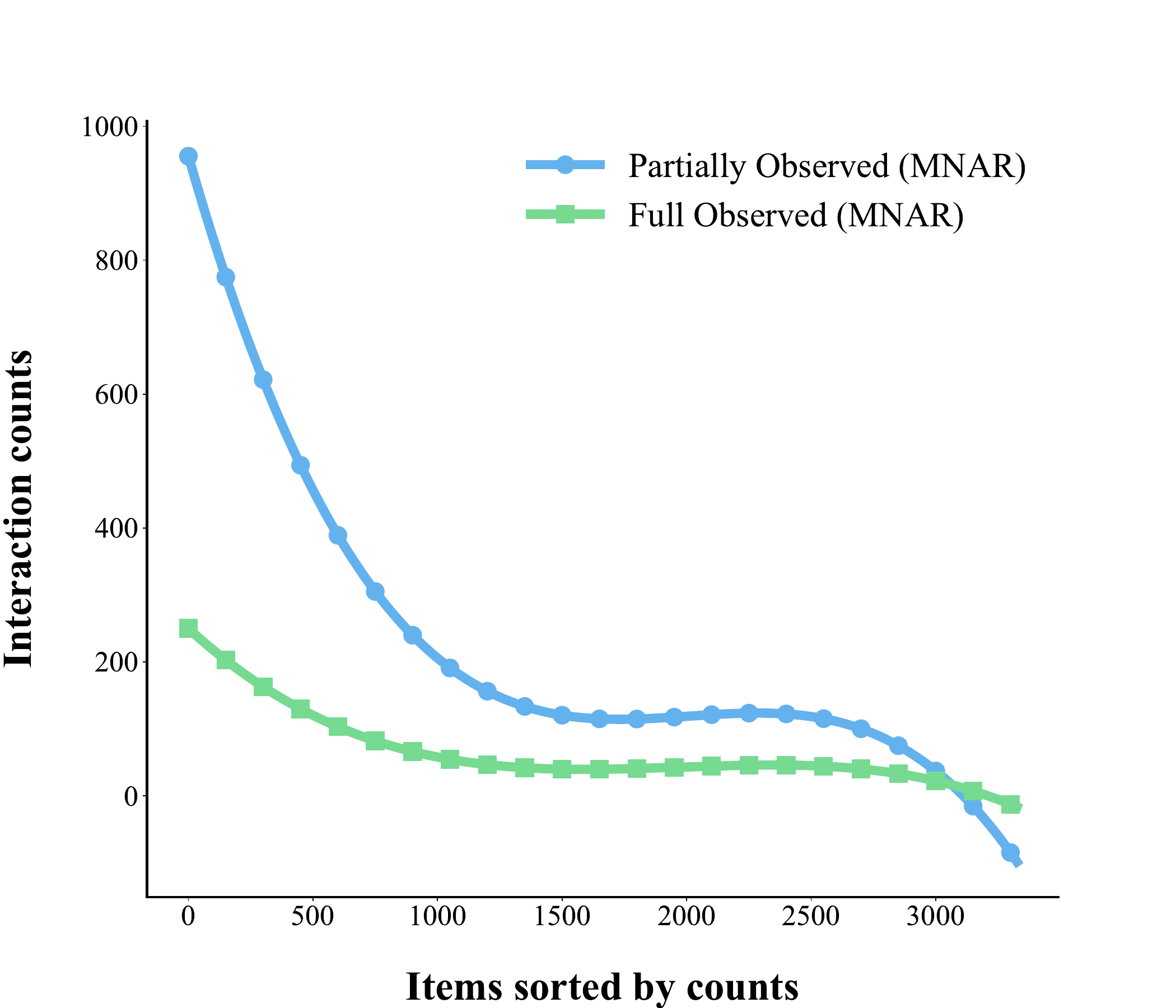}
\centerline{(b) KuaiRec}
\end{minipage}
\caption{The unknown exposure mechanism under different length feedback loops.}
\label{fig:2}
\end{figure}
The results show clear differences between the full observed data and the partial-observed data in Coat and Kuairec. The difference between the full observed data and the partial-observed data in Figure \ref{fig:2} indicates that the exposure mechanism does have an impact on the user-item interaction, thus most of the datasets in real scenario are long-tailed distributions. Otherwise, MAR data is not affected by feedback loops, and items in MAR data have nearly the same ratio, meaning that items have the same probability of being exposed to users, as shown in Figure \ref{fig:2} (a). The MNAR data is affected by both feedback loops and exposure mechanisms, so the items have different exposure probabilities and show significantly different trends with respect to the MAR data. However, the KuaiRec dataset, which suffers more from feedback loops, shows a high similarity between the full observed data and the partial-observed data, with a significantly higher ratio of head-to-tail items in partial-observed data compared to Coat. The results between the two types of datasets justify the above assumption, the exposure mechanism suffers from feedback loops.
\subsection{Exposure mechanism}

The exposure mechanism plays an indispensable role in the user-item interaction process by determining the exposure probability of the item to the user, which we can describe as {\small $\bm{O}(i)$}. Interaction occurs when the user is interested in the item and observes it:
\begin{equation}
\begin{small}
\begin{aligned} 
P\left(\bm{S}_{u,i} = 1\right) = P\left(\bm{R}_{u,i} = 1\right) \cdot \bm{O}(i),\\
\end{aligned}
\label{eq:2}
\end{small}
\end{equation}
{\small $P(\bm{S}_{u,i} = 1)$} indicates that the user rates the item, which means that an interaction has occurred. If the exposure mechanism is not affected by recommendation systems, such as the system makes random recommendations, which means feedback data obey MAR assumption, and items have the same opportunity to be exposed to users, leading to:
\begin{small}
\begin{equation}
\begin{aligned} 
\bm{O}(i) = \frac{1}{N}.
\end{aligned}
\label{eq:3}
\end{equation}
\end{small}

However, most recommendation systems make personal recommendations to improve the quality of recommendations and the user experience, which inevitably affects the exposure mechanism. To measure the change in the exposure mechanism, we can derive the exposure probability from the feedback data because we theoretically have access to all the data. That is:
\begin{equation}
\begin{small}
\begin{aligned} 
\bm{O}(i) &= \frac{\sum_{u = 1}^{M}\bm{S}_{u,i}}{\sum_{i = 1}^{N}\sum_{u = 1}^{M} \bm{S}_{u,i}}\\
&= \frac{\sum_{u = 1}^{M}P\left(\bm{S}_{u,j}= 1\right)}{\sum_{j = 1}^{N}\sum_{u = 1}^{M} P\left(\bm{S}_{u,j} = 1\right)}.\\
\end{aligned}
\label{eq:4}
\end{small}
\end{equation}
Since our problem is focused on the exposure mechanism under the feedback loops, we further transform Equation \ref{eq:4} into a more general formulation of the pattern. The current exposure mechanism at T = $t$ is mostly depend on the relevance score predicted by the current model and the past exposure mechanism at T = $t-1$.
\begin{equation}
\begin{small}
\begin{aligned} 
\bm{O}^{t}(i) = \frac{\sum_{u = 1}^{M}P\left(\bm{R}_{u,i}= 1\right)\cdot\bm{O}^{t - 1}(i)}{\sum_{j = 1}^{N}\sum_{u = 1}^{M} P\left(\bm{R}_{u,j}= 1\right)\cdot\bm{O}^{t - 1}(j)}.\\
\end{aligned}
\label{eq:5}
\end{small}
\end{equation}
Based on Equation \ref{eq:5}, we can calculate the difference of the exposure mechanisms after one feedback loop.
\begin{equation}
\begin{small}
\begin{aligned} 
\Delta \bm{O}^{t}(i) &= \bm{O}^{t}(i) - \bm{O}^{t - 1}(i) = \frac{\sum_{u = 1}^{M}P\left(\bm{R}_{u,i}= 1\right) - C}{C} \cdot \bm{O}^{t - 1}(i),\\
&where \quad C = \sum_{j = 1}^{N}\sum_{u = 1}^{M} P\left(\bm{R}_{u,j}= 1\right).
\end{aligned}
\label{eq:6}
\end{small}
\end{equation}
From Equation \ref{eq:6}, we can learn that items with high popularity can have low $\Delta \bm{O}^{t}(i)$, which means that the exposure probability will be higher compared to unpopular items. This can better explain why mainstream models suffer from the "rich-get-richer” problem. For ease of calculation, we extract the relationship between variation in exposure mechanism and static data.
\begin{equation}
\begin{small}
\begin{aligned} 
\Delta \bm{O}^{t}(i) \propto \sum_{u = 1}^{M}P\left(\bm{R}_{u,i}= 1\right).
\end{aligned}
\label{eq:7}
\end{small}
\end{equation}
If the data is collected from a recommendation system that makes random recommendations or through other methods such as surveys, we can regard it as free from the feedback loops, which means $T = 0$. Then we can have an iterative computation between $\bm{O}^{t}(i)$ and $\bm{O}^{t - 1}(i)$:
\begin{equation}
\begin{small}
\begin{aligned} 
\bm{O}(i) &= \frac{1}{N} \quad when\ T = 0,\\
\bm{O}^{t}(i) &= (1 + \sum_{u = 1}^{M}P\left(\bm{R}_{u,i}= 1\right)) \bm{O}^{t - 1}(i).
\end{aligned}
\label{eq:8}
\end{small}
\end{equation}
By Equation 8, we can calculate the $\bm{O}^{t}(i)$ at $T = t$ with no known exact exposure mechanism in the current timestamp, which can be formulated as:
\begin{equation}
\begin{small}
\begin{aligned} 
\bm{O}^{t}(i) = (1 + \sum_{u = 1}^{M}P\left(\bm{R}_{u,i}= 1\right))^{\alpha}.
\end{aligned}
\label{eq:9}
\end{small}
\end{equation}
where $\alpha$ is the frequency weight of the exposure mechanism that suffers from feedback loops. However, direct computation of relevance score between user-item pairs is time-consuming and the accuracy of the computed results is difficult to guarantee, so we need to find a replacement. From Equation \ref{eq:2}, we can obtain an inequality between the relevance score and the observed ratings:
\begin{equation}
\begin{small}
\begin{aligned} 
P\left(\bm{R}_{u,i}= 1\right) &\geq P\left(\bm{S}_{u,i}= 1\right) \quad Equal\ when\ \bm{O}^{t - 1}(i) = \frac{1}{N}.
\end{aligned}
\label{eq:10}
\end{small}
\end{equation}
Then the lower bound on $\bm{O}^{t}(i)$ can be directly computed from the observed ratings without external information:
\begin{equation}
\begin{small}
\begin{aligned} 
\bm{O}^{t}(i) &\geq \left(1 + \sum_{u = 1}^{M}P\left(\bm{S}_{u,i}= 1\right)\right)^{\alpha}.
\end{aligned}
\label{eq:11}
\end{small}
\end{equation}
The calculation method of the lower bound is understandable, existing methods such as UBPR and Rel-MF typically exploit the interaction counts of items as the main source of propensity scores and achieve better performance, but they focus on the current timestamp and ignore the variability of exposure mechanisms. Otherwise, Eq. 11 tells us that it is fallacious to simply use the interaction counts of the items, and therefore it is suboptimal to compute the propensity score based on the interaction counts.
\begin{figure}[tp]
\centering
\includegraphics[scale = 0.5]{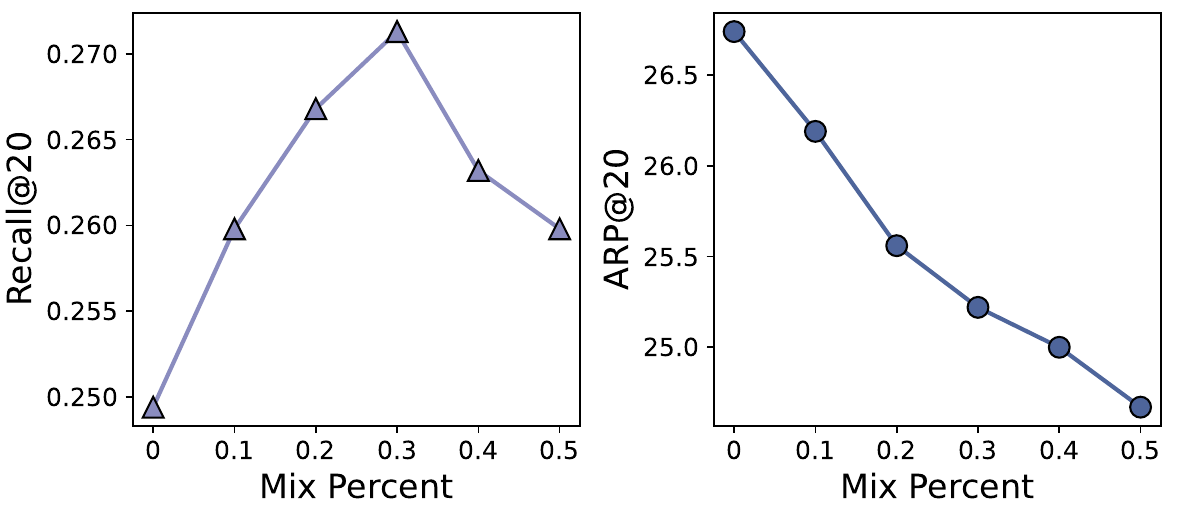}
\caption{Performance between different Exposure mechanisms, a higher value of the mixing percentage indicates a higher proximity to the original exposure mechanism.}
\label{fig:3}
\end{figure}
\par
This finding makes us aware that the exposure mechanism suffers from feedback loops, which will harm the new model. So the key to unbiased recommendation is to free the model from the loops. We conduct an experiment on Coat to prove that the change in the exposure mechanism does improve the quality of the recommendation. The performance of the backbone model MF in different exposure mechanisms is presented in Figure \ref{fig:3}. We mix various percentages of MAR data into the training set (MNAR data) to represent different exposure mechanisms. Increasing the percentage of mixing implies a closer fit to the original exposure mechanism. \par
As shown in Figure 3, the performance of the model does improve at the beginning, which proves that the change in the exposure mechanism does improve the performance of the model. However, the performance of the model deteriorates as the exposure mechanism closes to the original exposure mechanism up to a point since the change in the exposure mechanism is not always negative. For example, the fashion of clothes changes over time, and the corresponding styles of clothes that are popular at a certain time are bound to be popular with the public, therefore the exposure mechanism that gives preferential attention to such clothes will improve the performance of the recommendation system and user satisfaction, which is clearly different from the original exposure mechanism. Moreover, the average popularity of the final recommended items keeps decreasing as the similarity between the exposure mechanism and the original exposure mechanism increases, which indicates that the original exposure mechanism benefits the recommendation of cold items. \par
In summary, the effects of unknown exposure mechanisms on model performance under feedback loops are confounded, and the key to solving the problem is to maintain its beneficial effects while eliminating its negative ones.
\subsection{Bias of pairwise ranking algorithm}
Pairwise ranking-based methods, which learn a user's preference for two items by maximizing the predicted score of a positive item, are always biased. Since the choice of positive and negative items depends on the observed interactions instead of the relevance of user-item. Previous studies \cite{ref4,ref5,ref6} have defined unbiasedness in terms of expectation, making the expectation of the loss equal to the ideal loss, which suffers from practical limitations, such as the high variance of the reweighted loss. Therefore, we propose a more universal definition of unbiasedness similar to \cite{ref7}:
\begin{definition}
A loss function $\mathcal{L}$ is unbiased if it optimizes the ranking of the predicted score of user-item pairs to that of the true relevance score:
\begin{equation}
\begin{small}
    \begin{aligned}
    \hat{s}_{u,i} > \hat{s}_{u,j} \Leftrightarrow r_{u,i} > r_{u,j}\quad when\ \mathcal{L}\ converges. 
    \end{aligned}
\end{small}
\label{eq:12}
\end{equation}
where $\hat{s}_{u,i}$ is the predicted score of the user-item pair, and the $r_{u,i}$ is the true relevance score.
\end{definition}
Most of the paired ranking methods derive from Bayesian Personalized Ranking (BPR) \cite{ref8}:
\begin{equation}
\begin{small}
\begin{aligned} 
\hat{\mathcal{L}}_{BPR} = - \frac{1}{|\mathcal{D}_{pair}|} \sum_{u,i,j \in \mathcal{D}_{pair}} \ln \sigma \left(\hat{s}_{u,i} - \hat{s}_{u,j}\right),
\end{aligned}
\label{eq:13}
\end{small}
\end{equation}
where {\small $\mathcal{D}_{pair} = \left\{\left(u,i,j\right) \mid \mathbf{S}_{u,i} = 1,\mathbf{S}_{u,j} = 0\right\}$ } is the set of all possible triplets for pairwise ranking. Then with Equation \ref{eq:13} we prove that BPR is biased.
\begin{proposition}
Pairwise ranking methods are biased.
\end{proposition}

\begin{proof}
    
Pairwise ranking encourages the predicted scores of positive pairs for every user to be higher than negative ones, which can be represented as: 
\begin{equation}
\begin{small}
\begin{aligned} 
\hat{s}_{u,i} - \hat{s}_{u,j} > 0.
\end{aligned}
\label{eq:15}
\end{small}
\end{equation}
According to the Equation \ref{eq:2}, we have:
\begin{equation}
\begin{small}
\begin{aligned} 
\ln P\left(\bm{S}_{u,i} = 1\right) = \ln P\left(\bm{R}_{u,i} = 1\right) + \ln \bm{O}(i).
\end{aligned}
\label{eq:16}
\end{small}
\end{equation}
By combining these two equations, we can rewrite the ranking:
\begin{equation}
\begin{small}
\begin{aligned}
\ln P\left(\bm{R}_{u,i}\!=\! 1\right) \!+\! \ln \bm{O}(i)\! -\! \left[\ln P\left(\bm{R}_{u,j} \!=\! 1\right)\! +\! \ln \bm{O}(j) \right] \!>\! 0.\\
\end{aligned}
\label{eq:17}
\end{small}
\end{equation}
In contrast to the definition, Equation \ref{eq:13} does not serve the purpose of reflecting the true relevance score. This is because if the item i is more popular, as shown in Equation \ref{eq:7}, 
 $\textbf{O}(i)$ may be much higher than $\textbf{O}(j)$. Therefore, the pairwise wanking methods are biased, models optimized with them will produce biased recommendations.
\end{proof}

\section{Proposed Method: \textit{DPR} and \textit{UFN}}
In this section, we propose a new pairwise ranking method called DPR and theoretically demonstrate its unbiasedness. Then we introduce a plugin named Universal Anti-False Negative, called UFN, to alleviate the false negative problem in pairwise ranking.
\subsection{Dynamic personalized ranking}
We begin by presenting the construction of DPR:
\begin{equation}
\begin{small}
\begin{aligned} 
\mathcal{L}_{DPR} &= - \frac{1}{|\mathcal{D}_{pair}|} \sum_{u,i,j \in \mathcal{D}_{pair}} \ln \sigma \left[\frac{\hat{s}_{u,i}}{\gamma_{i}} - \frac{\hat{s}_{u,j}}{\gamma_{j}}\right],\\
\gamma_{i} &= \left(1 + \sum_{u = 1}^{M}P\left(\bm{S}_{u,i}= 1\right)\right)^{\alpha}.
\end{aligned}
\label{eq:18}
\end{small}
\end{equation}
where ${\gamma_{i}}$ is the exposure probability of item i, which can be used as a proxy for the probability of exposure according to Equation \ref{eq:11}; $\alpha$ is a hyper-parameter that reflects the length of the feedback loops suffered by the data, a higher value of $\alpha$ means longer loops. When $\alpha = 0$, DPR is equivalent to BPR because there are no feedback loops on the data, and the presence of the exposure mechanism does not contribute to bias. The value of the $\alpha$ depends on the training data.
\subsection{Variance control of DPR}
Unbiased pairwise methods commonly use inverse propensity scores of items as weights, such as UBPR and PDA, and lower accuracy in computing the propensity scores might lead to higher variance and sub-optimal recommendations, especially for the tail items with low exposure probability. To alleviate this problem, existing methods \cite{mfdu,ref4} apply non-negative estimators that clip large negative values for practical variance control but introduce additional bias.\par
Unlike existing methods, DPR uses the lower bound of the exposure mechanism as the weight, and once $\alpha$ is determined, the accuracy of the lower bound is guaranteed. Equation \ref{eq:17} shows that $\gamma_{i}$ is bounded at $[1,4]$ and $\frac{1}{\gamma_{i}}$ is therefore bounded at $[0.25,1]$. Due to the large number of items, the final value will tilt towards the left side of the range, so that the variance of DPR is acceptable even for the least popular items. In summary, the variance of DPR is much lower than the existing re-weighting methods without any additional operations.
\subsection{Unbias of DPR}
Similar to BPR, DPR also encourages the score of positive pairs to be higher than negative ones. DPR can also be reformulated as:
\begin{equation}
\begin{small}
\begin{aligned} 
\frac{\hat{s}_{u,i}}{\gamma_{i}} - \frac{\hat{s}_{u,j}}{\gamma_{j}} > 0.\\
\end{aligned}
\label{eq:19}
\end{small}
\end{equation}

\begin{proposition}
DPR is unbiased under the assumption of Equation \ref{eq:5}.
\end{proposition}

\begin{proof}
According to the assumption Equation \ref{eq:5} and Equation \ref{eq:12}, given the pairs {\small $(u,i,j) \in \mathcal{D}_{pair}$}, we have:
\begin{equation}
\begin{small}
\left\{
\begin{aligned} 
&\ln P(\bm{R}_{u,i} = 1) + \ln \bm{O}(i) - \ln \gamma_{i},\\
&\ln P(\bm{R}_{u,j} = 1) + \ln \bm{O}(j) - \ln \gamma_{j}.
\end{aligned}
\right.
\label{eq:20}
\end{small}
\end{equation}
Substitution for  Equation \ref{eq:16} gives:
\begin{equation}
\begin{small}
\begin{aligned} 
\ln P(\bm{R}_{u,i} = 1) &- \ln P(\bm{R}_{u,j} = 1) > 0 \\
\Leftrightarrow \quad\quad\quad\quad\quad r_{u,i} &- r_{u,j} > 0.\\
\end{aligned}
\label{eq:21}
\end{small}
\end{equation}
If a proper parameter $\alpha$ is chosen, the negative impact of the exposure mechanism is canceled by $\gamma_{i}$ and $\gamma_{j}$. Model prediction scores for user-item pairs are ranked in the same order as the true relevance scores. Thus, the DPR is unbiased.

\end{proof}

\begin{algorithm}[t]
\caption{ Dynamic Personalized Ranking (DPR)}
\label{alg:Framwork}
\begin{algorithmic}[1]
\REQUIRE ~~\\
    The set of observed interaction data $\{S_{u,i}\}$;\\
    Learning rate $\mu$,batch size $B$,regularization parameter $\lambda$;\\
    Exposure parameter $\alpha$,false negative parameter$\beta$;
\ENSURE ~~\\
    learned model parameters;
    \STATE Initialize the model parameters $\textbf{U},\textbf{V}$, and calculate the exposure probability of items ${\gamma_{i}}$ ;
    \STATE Sample negative samples for pairs (u, i) sorted by current score $\textbf{S}_{u,j}$ for composition $\mathcal{D}_{pair}$;
    \REPEAT
    \STATE Sample a size $B$ of mini-batch data from $\mathcal{D}_{pair}$;
    \STATE Compute $\mathcal{L}_{DPR}$;
    \STATE Update model learnable parameters $\textbf{U},\textbf{V}$ ;
    \STATE Resampling the negative samples in the same way;
    \UNTIL convergence
\RETURN $\textbf{U},\textbf{V}$;
\end{algorithmic}
\end{algorithm}
\subsection{False negative samples approach}
The pairwise ranking method usually assumes that unobserved interactions are all negative, which is unrealistic because an item may not be exposed to the user, rather than the user does not prefer it, thus introducing false negative samples into the training process. The existence of false negative samples hurts the performance of the model because it prevents the model from fully tapping into user preferences. Previous studies \cite{ref9,ref10} have shown that both false negative samples and hard negative samples have large scores, but false negative samples have comparatively lower prediction variance. This means that false negative samples maintain a high prediction score in the training process. On the basis of these findings, we introduce a plugin named Universal Anti-False Negative (UFN) to alleviate the false negative problem:
\begin{equation}
\begin{small}
\begin{aligned} 
UFN(\hat{s}_{u,j}) = \left(1 - f(\hat{s}_{u,j})\right)^{\beta},
\end{aligned}
\label{eq:22}
\end{small}
\end{equation}
where $\hat{s}_{u,j}$ is the predicted score of negative items, and $\beta$ is the hyper-parameter to control the strength, with a higher $\beta$ implying a stronger effort to alleviate false negative samples. $f(\hat{s}_{u,j})$ is a mapping function like $sigmod(\cdot)$, for the convenience of the proof, we use $tanh(\cdot)$.  The DPR with UFN can be written as:
\begin{equation}
\begin{small}
\begin{aligned} 
\mathcal{L}_{DPR}\! =\! - \frac{1}{|\mathcal{D}_{pair}|}\! \sum_{{\tiny u,i,j \in \mathcal{D}_{pair}}} \!\ln \sigma\! \left[\frac{\hat{s}_{u,i}}{\gamma_{i}}\! -\! \frac{UFN(\hat{s}_{u,j})}{\gamma_{j}}\hat{s}_{u,j}\right].
\end{aligned}
\label{eq:23}
\end{small}
\end{equation}

\begin{proposition}

UFN can alleviate the false negative problem
\end{proposition}
\begin{proof}

For obvious results, the next proof procedure will be based on matrix factorization, which is common in recommendation systems. It is common to use the regularization algorithm in the training process, so DPR with regularization can be rewritten as:
\begin{equation}
\begin{small}
\begin{aligned}
\underset{\Theta}{\arg\min}\! \sum_{{\tiny u,i,j \in \mathcal{D}_{pair}}}\! \left\{\! \lambda\|\Theta\|^2 \!-\! \ln \!\sigma\! \left[\frac{\hat{s}_{u,i}}{\gamma_{i}}\!- \frac{UFN(\hat{s}_{u,j})}{\gamma_{j}}\hat{s}_{u,j}\right] \!\right\}.\nonumber
\end{aligned}
\end{small}
\end{equation}
The gain after a single batch is obtained by SGD is:
\begin{equation}
\begin{small}
\left\{
\begin{aligned}
&\frac{\partial \mathcal{L}}{\partial u} \!= \!\mathbf{p_{j}}\! \left(1 - tanh(\hat{s}_{u,j}) \right)\! -\! \mathbf{p_{i}} \!+\! \left(\lambda\! +\! (1\! -\! tanh(\hat{s}_{u,j})^2 )\|\mathbf{p_{j}}\|^2\right)\mathbf{p_{u}},\\
\\
&\frac{\partial \mathcal{L}}{\partial j}\! =\! \mathbf{p_{u}} \!\left(1 - tanh(\hat{s}_{u,j}) \right)\! +\! \left(\lambda \!+\! (1\! -\! tanh(\hat{s}_{u,j})^2 )\|\mathbf{p_{u}}\|^2\right)\mathbf{p_{j}}.\\
\end{aligned}
\right.
\label{eq:24}
\end{small}
\end{equation}
where {\small $\mathbf{p_{u}},\mathbf{p_{i}}$} and {\small$\mathbf{p_{j}}$} is the latent factor of the model parameters {\small$\bm{U},\bm{V}$}. As we can see from Equation \ref{eq:21} when a higher prediction score of false negative sample j contributes less to the user u, which means that false negative samples do not damage the model construction of user preferences. In addition, the plugin helps false negative samples to stable their preferences, which proves its effectiveness and improves the quality of recommendations.
\end{proof}
\subsection{Comparison with IPS-based methods}
IPS and DPR both focus on exposure mechanisms to achieve unbiased recommendations, but with different views of the exposure mechanism and thus provide different solutions. Here we briefly explain the difference between DPR and IPS in modeling the exposure mechanism and the superiority of our novel design. UBPR is the application of IPS in pairwise ranking:
\begin{equation}
\begin{small}
\begin{aligned}
\frac{1}{\|\mathcal{D}_{pair}\|} \sum_{u,i,j \in \mathcal{D}_{pair}} \frac{\bm{S}_{u,i}}{\theta_{u,i}}\left(1 - \frac{\bm{S}_{u,i}}{\theta_{u,i}}\right) \ln\sigma(\hat{s}_{u,i}-\hat{s}_{u,j}\big). \nonumber
\end{aligned}
\end{small}
\end{equation}

The problems of IPS can be summarized as follows:\par
\begin{itemize}
\item IPS believes that the exposure mechanism is constant, but the fact is that it changes over loops.\par
\item IPS substitutes propensity score $\theta_{u,i}$ for the exposure mechanism, which is an oversimplification of the exposure mechanism. In addition, whether the re-weight loss is equal to the ideal expectation as the condition for unbiased or not will suffer from high variance if the propensity score is incorrect.\par
\end{itemize}
\par
Our DPR alleviates the following problems:\par
\begin{itemize}
\item We conduct an in-depth analysis of the exposure mechanism to find its key factors over time, which are used in the DPR to counteract the effects of the exposure mechanism.\par
\item DPR considers changes in the exposure mechanism over time and uses a parameter $\alpha$ to improve its generalization across different datasets.\par
\item Compared to using the ideal expectation to achieve unbias in IPS, DPR focuses on discovering the true performance of users and theoretical proof of its validity.\par
\item IPS-based methods need to compute the propensity score of items and thus suffer from high variance problems. DPR uses a lower bound on the exposure probability and thus has a stable performance compared to IPS-based methods.
\end{itemize}
\section{Experiments}
In this section, we first describe the experimental setup and then demonstrate the effectiveness of the proposed DPR and UFN.
\subsection{Experimental setup}
\subsubsection{Datasets.}Our experiments are conducted on six real-world datasets from three datasets have full observed data Yahoo! R3\footnote{\href{https://webscope.sandbox.yahoo.com/}{https://webscope.sandbox.yahoo.com/}}, Coat\footnote{\href{https://www.cs.cornell.edu/ schnabts/mnar/}{https://www.cs.cornell.edu/ schnabts/mnar/}} and KuaiRec\footnote{\href{https://chongminggao.github.io/KuaiRec/}{https://chongminggao.github.io/KuaiRec/}} \cite{kuai}, and three common dataset Movielens 100K\footnote{\href{https://grouplens.org/datasets/movielens/}{https://grouplens.org/datasets/movielens/}} \cite{ml-100k}, Movielens 1M and Lastfm\footnote{\href{https://grouplens.org/datasets/hetrec-2011/}{https://grouplens.org/datasets/hetrec-2011/}} \cite{lastfm}. Following prior works, ratings in Yahoo! R3, Coat, Movielens 100K, Movielens 1M and  Lastfm are binarized by setting ratings under 4 to 0 and the rest to 1. For KuaiRec, We define the interest of a user-item interaction as the watching ratio, which is the ratio of watching time length to the total length, and also binarize it by setting the ratio under 2 to 0, and the rest to 1. Meanwhile, we apply leave-one-out to split training, validation, and test data for all datasets. The properties of datasets are summarized in Table \ref{tab:3}.
\begin{table}[tp]
\renewcommand\arraystretch{1.25}
    \centering
    \caption{Datasets Statistics.}
    \label{tab:3}
    \begin{tabular}[b]{crrrrc}
    \bottomrule
    Dataset & \#Users & \#Items & \#Interactions & \textit{Sparsity} \\
    \bottomrule
    \midrule
    Coat & 290 & 300 & 6,960 & 92.0\%\\
    Yahoo! R3 & 15,400 & 1,000 & 311,704 &97.9\%\\
    KuaiRec & 7,176 & 10,729 & 12,530,806 & 83.7\%\\
    Movielens 100K & 943 & 1,682 & 100,000 & 93.6\%\\
    Lastfm & 1,891 & 12,522 & 149,959 & 99.2\%\\
    Movielens 1M & 6,039 & 3,705 & 802,553 & 95.9\%\\
    \bottomrule
    \end{tabular}
\end{table}
\subsubsection{Simulation dataset} To evaluate whether the effect of the feedback loops is indeed present, we conduct experiments on simulated datasets. The simulation dataset contains 200 users and 500 items, and we generate the initial dataset by randomly exposing 20 items per user and 20 users per item. After initialization, we randomly select k = 2 of the top 10 items recommended by the model in the current loop as the new user-item interaction. We run a total of 50 loops for each method and then summarize the results. The initial data static are shown in Figure \ref{fig:4}.
\begin{figure}
    \centering
    \begin{minipage}[t]{0.48\linewidth}
    \centering
    \centerline{\includegraphics[width = \linewidth]{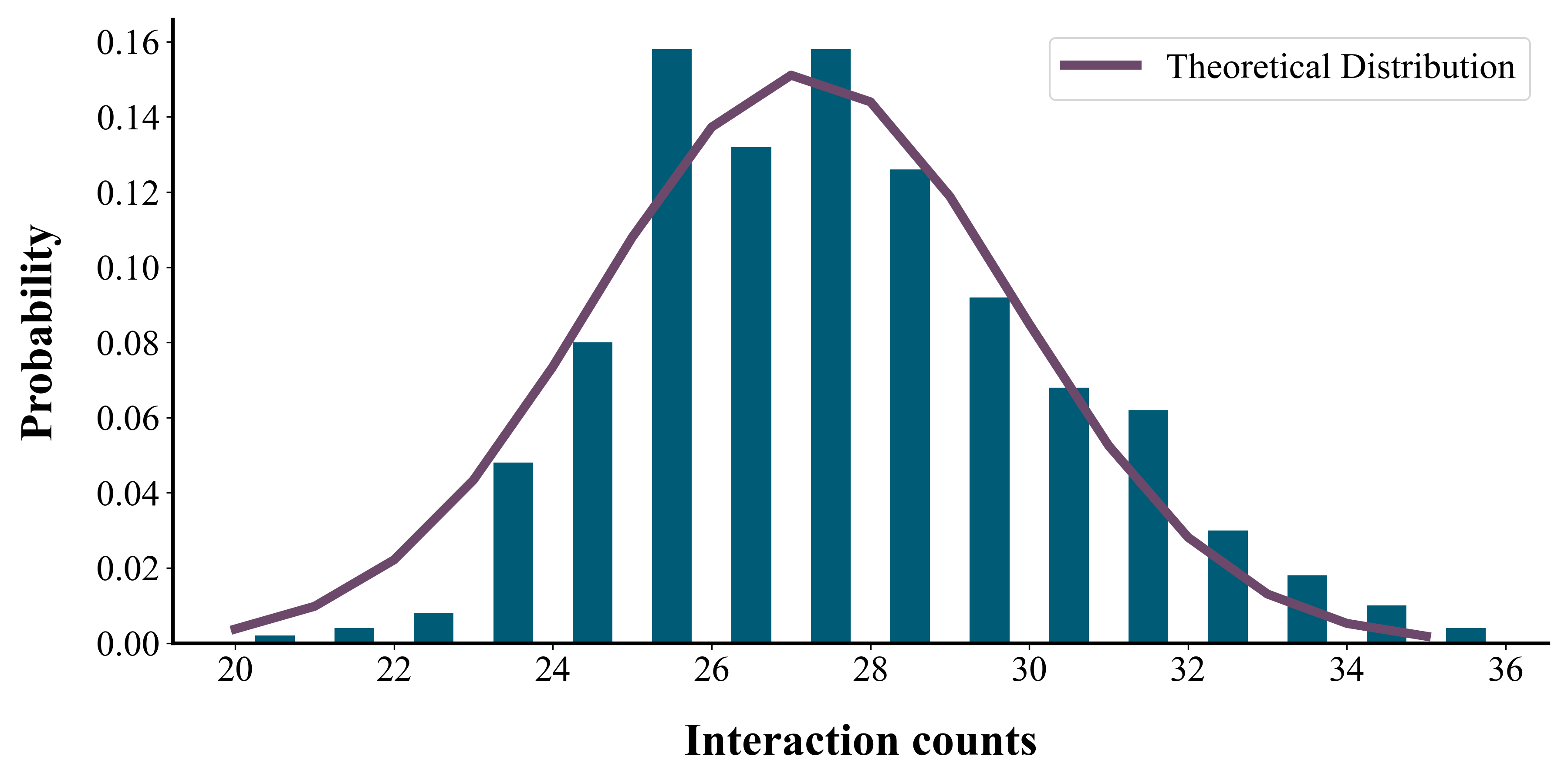}}
    \centerline{(a)}
    \end{minipage}
    \hfill
    \begin{minipage}[t]{0.48\linewidth}
    \centering
    \centerline{\includegraphics[width = \linewidth]{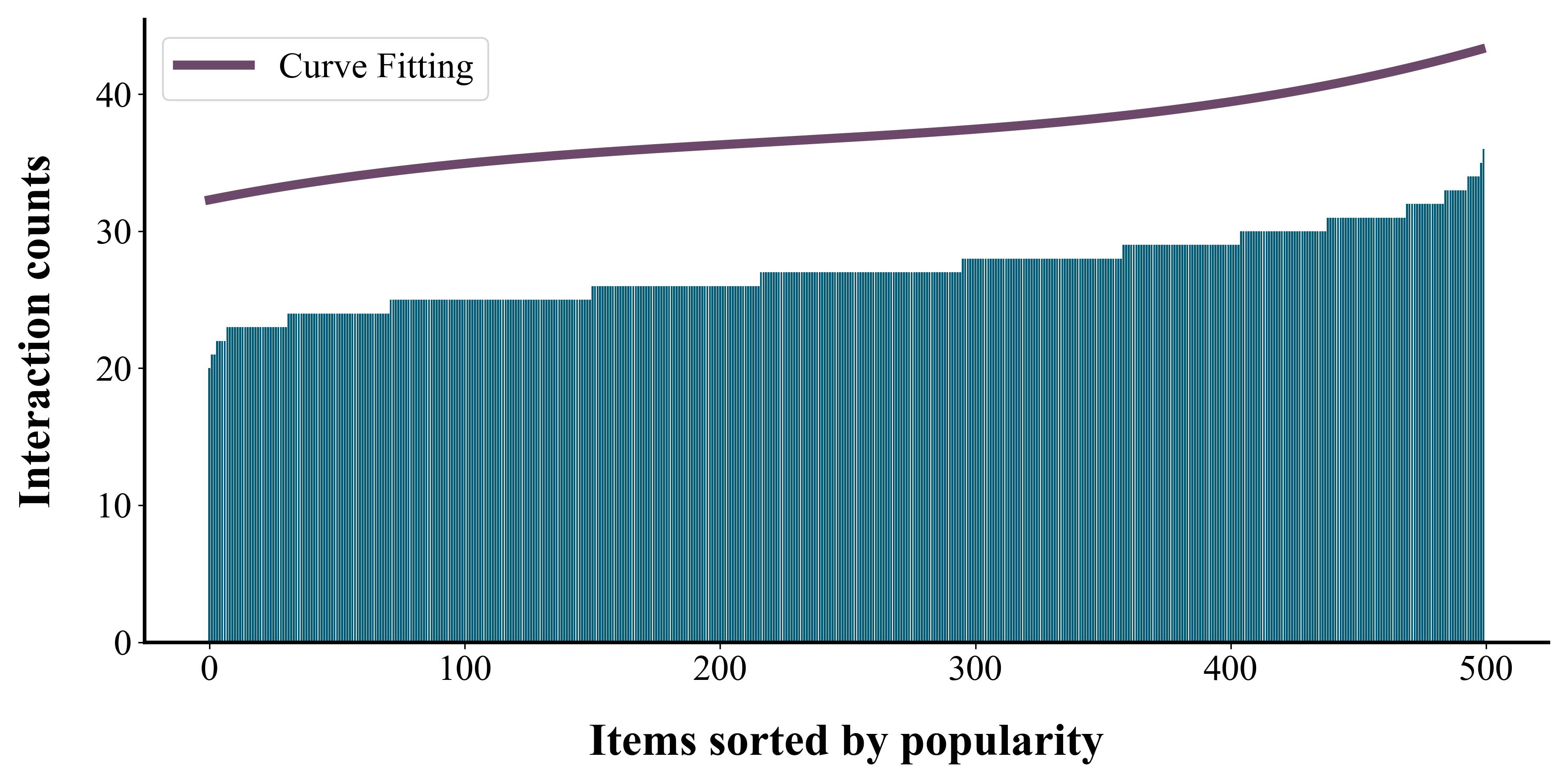}}
    \centerline{(b)}
    \centering
    \end{minipage}
    \caption{Simulation Dataset Statistic}
    \label{fig:4}
\end{figure}
\subsubsection{Evaluation.} 
Note that the sampling-based evaluation approach does not truly reflect the ability of the model to capture the true preferences of users, simply fitting the data may also have better performance. To this end, we report the all-ranking performance w.r.t. two widely used metrics: Recall, and NDCG cut at K = 5. In order to measure the popularity of recommended items, we use average popularity rank (ARP) and tail percentage (TAP) as the other two indicators, and the formula is:
\begin{equation}
\begin{aligned}
\mathrm{Average Popularity Rank (ARP)@K}=\frac{1}{|U|} \sum_{u \in U } \frac{\sum_{i \in R_{u}} \phi(i)}{|R_{u}|} \nonumber
\end{aligned}
\end{equation}
\begin{equation}
\begin{aligned}
\mathrm {Tail Percentage (TAP)@K}=\frac{1}{|U|} \sum_{u \in U} \frac{\sum_{i \in R_{u}} {\delta(i \in T)}}{|R_{u}|} \nonumber
\end{aligned}
\end{equation}
where $\phi(i)$ is the ranking of times item i has been rated in the training set, $\delta(i \in T)$ is the tail item set, we set the last 80\% of the total rank of the popularity of all items as the tail item. $\bm{R}_{u}$ is the recommended list of items for user $u$. We apply sampling-based evaluation on hyperparameter experiments, which can better reflect the performance of the model in terms of modeling accuracy. Specifically, we sample 99 negative samples for each user-item interaction in the test set.
\subsubsection{Baselines}
We also use BPR and six other debiasing methods for comparison, including four pairwise methods and two pointwise methods.
\begin{itemize}
\item \textbf{BPR} \cite{ref8} is the basic pairwise ranking algorithm that is widely used in training. 
\item \textbf{UBPR} \cite{ref4} is an IPS-based method that modifies pairwise loss. 
\item \textbf{EBPR} \cite{ref12} uses item-based explanations to enhance the interpretability of the pairwise loss. 
\item \textbf{PDA} \cite{pda} is an method focus on eliminate popularity problem.
\item \textbf{UPL} \cite{upl} is an unbiased method that focused on the high-variance problem of IPS-based methods.
\item \textbf{Rel-MF} \cite{ref15} is an unbiased method of binary cross-entropy loss for solving missing-not-at-random problems.
\item \textbf{CPR} \cite{ref7} constructs special training samples to achieve debiasing. 
\item \textbf{MFDU} \cite{mfdu} simultaneously eliminates the bias of clicked and unclicked data.
\end{itemize}
For a fair comparison, we apply the above methods to the widely used model, Matrix Factorization (MF) \cite{ref26} and LightGCN \cite{ref27}.
\subsubsection{Implementation details} We implement DPR and baselines in PyTorch. We take the embedding dimension as 64, and all models are trained with the Adam optimizer via early stopping. We set the learning rate to 1e-3 and the $l_2$-regularization weight to 1e-6. We randomly sample 10 negative items for each positive item for pairwise ranking methods. For DPR, we tune the hyper-parameter $\alpha$ in the range of $[0,6]]$ for different data sets, with the default setting $\beta = 1$. We implement DPR$^-$ as DPR without the use of the UFN to measure the effectiveness of the proposed plugin. To detect significant differences in DPR and the best baseline on each data set, we repeat their experiments five times by varying the random seeds, we choose the average performance to report. All ranking metrics are computed at a cutoff K = 5 for the Top-5 recommendation\cite{xu02}.\par
\begin{figure}[tp]
    \begin{minipage}{0.9\linewidth}
        \centering
        \centerline{\includegraphics[width= 0.5\linewidth]{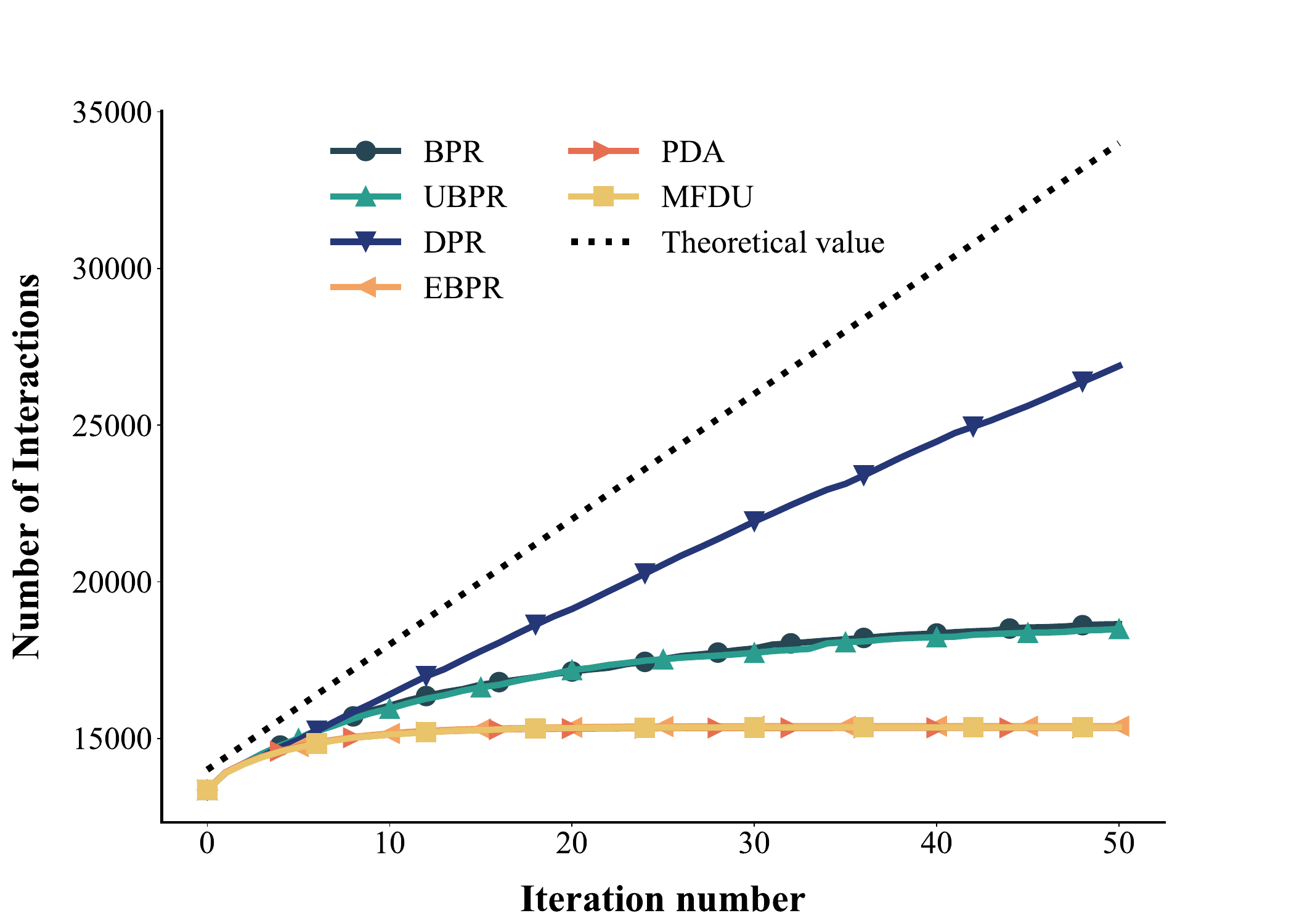}}
        \centerline{(a)}
        \centering
        \centerline{
        \centering
        \begin{minipage}{0.48\linewidth}
        \centering
        \centerline{\includegraphics[width=\linewidth]{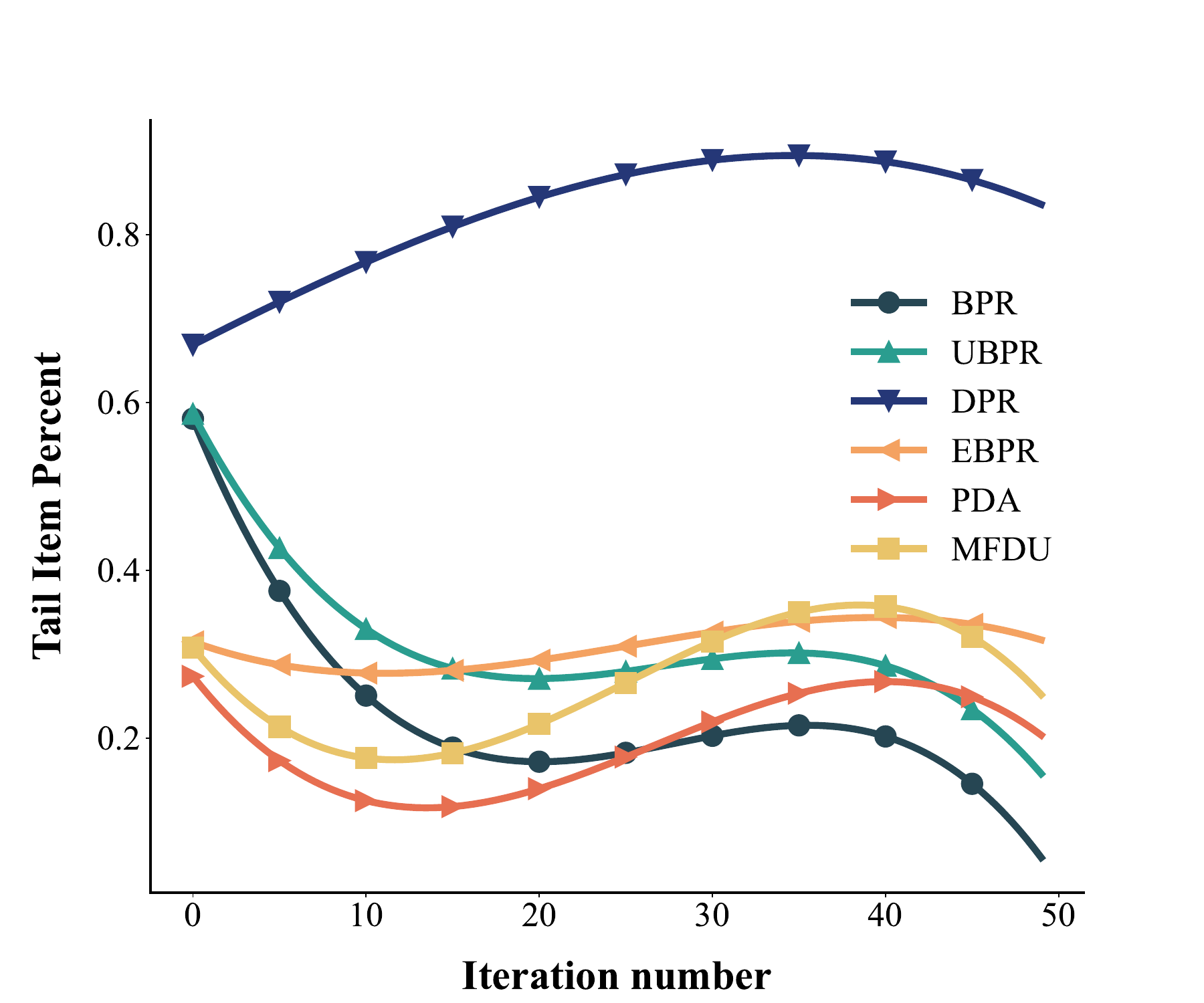}}
        \centerline{(b)}
        \end{minipage}
        
        \begin{minipage}{0.48\linewidth}
        \centerline{\includegraphics[width=\linewidth]{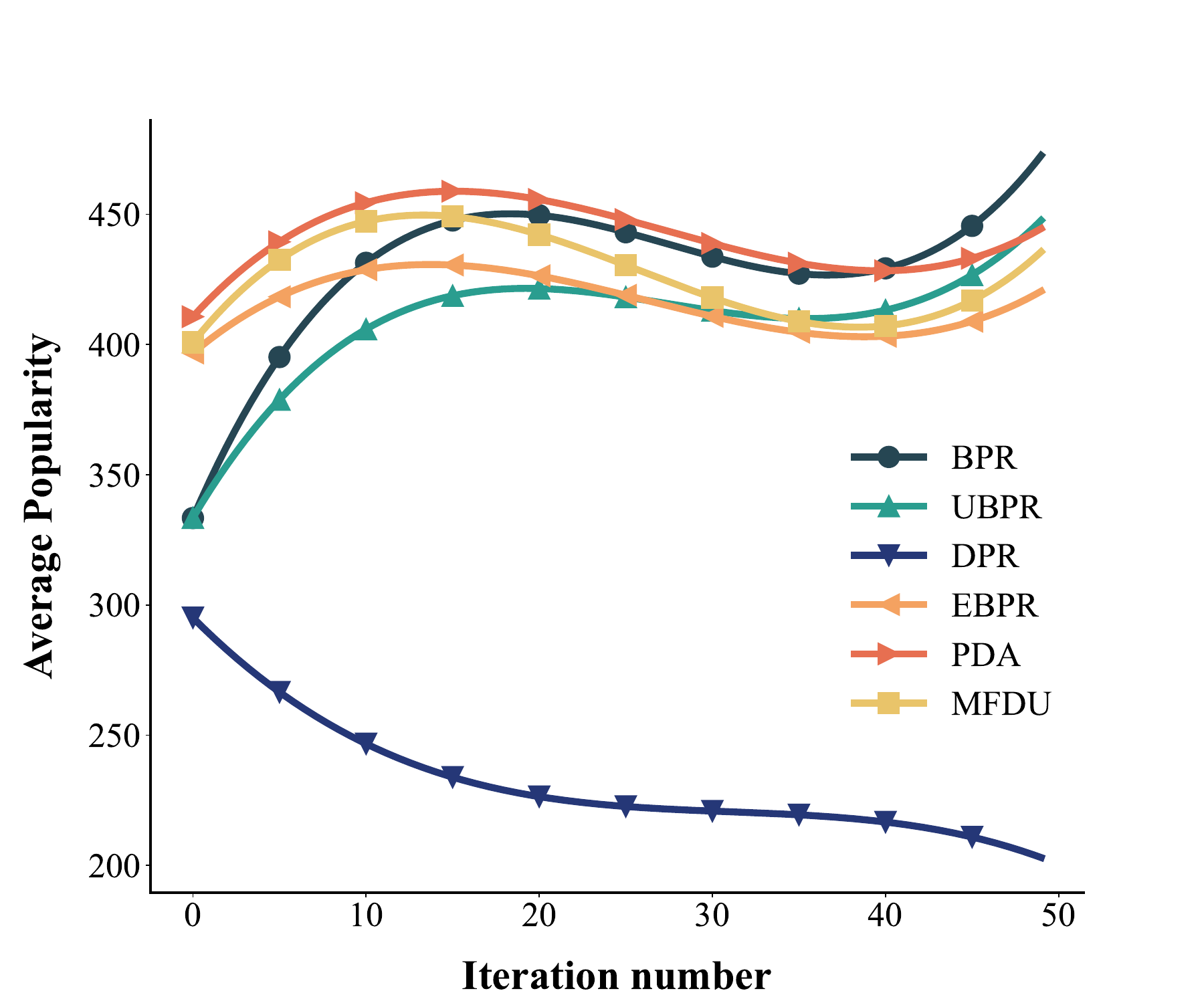}}
        \centerline{(c)}
        \end{minipage}
        }
    \end{minipage}%
    \caption{Simulation experiments results applying our framework on simulation dataset with MF. The recommendation performance is evaluated as a ranking task at K = 30. Higher Number of Interactions and Tail Item Percent mean better model performance; Average Popularity Rank the lower, the better.}
    \label{fig:5}
\end{figure}
Following \cite{mfdu,upl}, for MFDU, we separately computed the propensity scores for positive and negative items respectively as:
\begin{equation}
\begin{aligned}
    \theta_{*,i}^{+} = \left( \frac{n_i}{max_{i \in I}n_i}\right)^{0.5}, \nonumber \\
    \theta_{*,i}^{-} = \left(1 -  \frac{n_i}{max_{i \in I}n_i}\right)^{0.5}. \nonumber
\end{aligned}
\end{equation}
where $n_i$ means the number of interactions to the item i by all users. While for RelMF, UBPR, UPL, and PDA, we used $\theta_{*, i}^{+}$ as the propensity scores for all of the items. For DPR, we calculate $P\left(\bm{S}_{u,i}= 1\right)$ for all items as follows:
\begin{equation}
    P\left(\bm{S}_{u,i}= 1\right) = \frac{n_i}{max_{i \in I}n_i}. \nonumber
\end{equation}

\subsection{Simulation experiments}
We conduct experiments with all baselines and the proposed method DPR on simulated datasets, and the results are shown in Figure \ref{fig:5}. As can be seen from Figure \ref{fig:5} (a), as the feedback loops progress, there are almost no new user-item interactions in baselines after 30 loops. This means that users are surrounded by popular items and have no opportunity to come into contact with unpopular items, a phenomenon often referred to as filter bubbles. The user-item interaction growth rate of baselines is significantly below the theoretical value, and the baseline with better performance on debiasing may even be lower than the BPR, largely due to they seek to eliminate the bias in the current timestamp and ignore the impact of the feedback loops. DPR achieves a significant rate increase compared to other baselines, the existence of a gap between the growth curve of DPR and the theoretical value is justified by the fact that items cannot be preferred by all users and thus we cannot achieve a perfect curve like the theoretical value. The performance of DPR on the growth rate of user-item interactions validates that DPR can help the model break feedback loops and delay the arrival of filter bubbles.\par
DPR helps the model escape the feedback loop and perform better than the baselines in terms of debiasing. As can be seen from Figure \ref{fig:5} (b) and (a), DPR has a significantly high value in the percentage of tail items, while the baseline reaches a lower value. As the model iterates, the baselines fail to capture the variability of exposure mechanisms, such that items with high popularity have high exposure probabilities and deteriorate as the iteration progresses, thus achieving lower performance in terms of tail items. Debiasing loss functions such as UBPR, MFDU and EBPR perform better than BPR, demonstrating that debiasing methods can alleviate the bias accumulation in feedback loops, but the performance is not measurable in the long term as they formulate the problem in a single time step.

\begin{table}[tp]
\renewcommand\arraystretch{1.25}
    \caption{Overall performance of applying our framework on three full observed datasets with two recommendation models. The best result in each column is in bold; The best baseline result in each line is underlined. The recommendation performance is evaluated as a ranking task. Higher Recall and NDCG mean better model performance; ARP is a measure of the influence of feedback loops on the model, the lower, the better.}
    \label{tab:4}
\resizebox{\linewidth}{!}{
\begin{tabular}{@{}cccccccccccc@{}}
\toprule
                          &                            &                            & \multicolumn{3}{c}{Coat}                                                                                                 & \multicolumn{3}{c}{Yahoo}                                                                           & \multicolumn{3}{c}{Kuairec}                                                           \\ \midrule
                          &                            & \multicolumn{1}{c|}{}      & Recall@5                              & NDCG@5                                & \multicolumn{1}{c|}{ARP@5}               & \multicolumn{1}{c}{Recall@5} & \multicolumn{1}{c}{NDCG@5} & \multicolumn{1}{c|}{ARP@5}              & \multicolumn{1}{c}{Recall@5} & \multicolumn{1}{c}{NDCG@5} & \multicolumn{1}{c}{ARP@5} \\ \midrule
\multirow{7}{*}{MF}       & \multirow{5}{*}{Pairwise}  & \multicolumn{1}{c|}{BPR}   & 0.01702                               & 0.05498                               & \multicolumn{1}{c|}{210.40965}           & 0.02313                      & 0.01691                    & \multicolumn{1}{l|}{993.31318}          & 0.05655                      & 0.66579                    & 10607.80000               \\
                          &                            & \multicolumn{1}{c|}{UBPR}  & 0.01681                               & 0.05493                               & \multicolumn{1}{c|}{{\ul 189.01965}}     & {\ul 0.02637}                & 0.01733                    & \multicolumn{1}{l|}{970.32435}          & 0.05653                      & 0.66455                    & 10607.49270               \\
                          &                            & \multicolumn{1}{c|}{EBPR}  & 0.01681                               & 0.05422                               & \multicolumn{1}{c|}{255.27896}           & 0.02527                      & 0.01680                    & \multicolumn{1}{l|}{994.95318}          & 0.05655                      & 0.66374                    & 10607.80000               \\
                          &                            & \multicolumn{1}{c|}{PDA}   & 0.01735                               & 0.05651                               & \multicolumn{1}{c|}{255.15689}           & 0.02215                      & 0.01580                    & \multicolumn{1}{l|}{995.07494}          & 0.05655                      & 0.66579                    & 10607.80000               \\
                          &                            & \multicolumn{1}{c|}{UPL}   & 0.01757                               & 0.05690                               & \multicolumn{1}{c|}{265.37344}           & 0.02533                      & 0.01741                    & \multicolumn{1}{l|}{993.93612}          & 0.05129                      & 0.62390                    & 10607.50000               \\ \cmidrule(l){2-12} 
                          & \multirow{2}{*}{Pointwise} & \multicolumn{1}{c|}{RelMF} & 0.01832                               & {\ul 0.05944}                         & \multicolumn{1}{c|}{194.67310}           & 0.02512                      & {\ul 0.01861}              & \multicolumn{1}{l|}{987.01681}          & 0.05655                      & 0.66571                    & {\ul 10607.49135}         \\
                          &                            & \multicolumn{1}{c|}{MFDU}  & {\ul 0.01843}                         & 0.05867                               & \multicolumn{1}{c|}{284.49655}           & 0.02256                      & 0.01546                    & \multicolumn{1}{l|}{{\ul 993.05911}}    & {\ul 0.05655}                & {\ul 0.66579}              & 10607.80000               \\ \cmidrule(l){2-12} 
                          & \multirow{2}{*}{Ours}                       & \multicolumn{1}{c|}{$\rm DPR^{-}$}   & 0.01921                      & 0.06132                      & \multicolumn{1}{c|}{\textbf{174.64896}}  & 0.02760             & 0.01872           & \multicolumn{1}{l|}{980.11035} & 0.05695            & 0.66628          & 10504.79300      \\
                          &                                             & \multicolumn{1}{c|}{DPR}   & \textbf{0.02035}                      & \textbf{0.06199}                      & \multicolumn{1}{c|}{178.32158}  & \textbf{0.02831}             & \textbf{0.01901}           & \multicolumn{1}{l|}{\textbf{972.92581}} & \textbf{0.05732}            & \textbf{0.67028}          & \textbf{10454.62584}      \\ \midrule
\multirow{8}{*}{LightGCN} & \multirow{5}{*}{Pairwise}  & \multicolumn{1}{c|}{BPR}   & \multicolumn{1}{l}{0.01713}           & \multicolumn{1}{l}{0.05421}           & \multicolumn{1}{l|}{201.79896}           & 0.02158                      & 0.01450                    & \multicolumn{1}{l|}{902.34749}          & 0.00291                      & 0.02619                    & 10288.08072               \\
                          &                            & \multicolumn{1}{c|}{UBPR}  & \multicolumn{1}{l}{0.01713}           & \multicolumn{1}{l}{0.05369}           & \multicolumn{1}{l|}{198.06517}           & 0.01535                      & 0.01005                    & \multicolumn{1}{l|}{{\ul 798.96615}}    & 0.00282                      & 0.04705                    & 9607.58150                \\
                          &                            & \multicolumn{1}{c|}{EBPR}  & \multicolumn{1}{l}{{\ul 0.01767}}     & \multicolumn{1}{l}{{\ul 0.05610}}     & \multicolumn{1}{l|}{201.18551}           & {\ul 0.02470}                & {\ul 0.01896}              & \multicolumn{1}{l|}{958.38517}          & 0.01089                      & 0.10802                    & 10530.25301               \\
                          &                            & \multicolumn{1}{c|}{PDA}   & \multicolumn{1}{l}{0.01670}           & \multicolumn{1}{l}{0.05366}           & \multicolumn{1}{l|}{195.31827}           & 0.01892                      & 0.01400                    & \multicolumn{1}{l|}{936.49770}          & 0.00768                      & 0.05380                    & 10507.58242               \\
                          &                            & \multicolumn{1}{c|}{UPL}   & \multicolumn{1}{l}{0.01670}           & \multicolumn{1}{l}{0.05256}           & \multicolumn{1}{l|}{203.44310}           & 0.02156                      & 0.01455                    & \multicolumn{1}{l|}{853.06432}          & 0.00599                      & 0.10106                    & 9840.43940                \\ \cmidrule(l){2-12} 
                          & \multirow{2}{*}{Pointwise} & \multicolumn{1}{c|}{RelMF} & \multicolumn{1}{l}{0.01627}           & \multicolumn{1}{l}{0.05211}           & \multicolumn{1}{l|}{209.18172}           & 0.02021                      & 0.01470                    & \multicolumn{1}{l|}{854.45460}          & 0.00379                      & 0.05228                    & {\ul \textbf{9342.29305}} \\
                          &                            & \multicolumn{1}{c|}{MFDU}  & \multicolumn{1}{l}{0.01573}           & \multicolumn{1}{l}{0.05122}           & \multicolumn{1}{l|}{{\ul 196.58034}}     & 0.02232                      & 0.01583                    & \multicolumn{1}{l|}{956.23090}          & {\ul 0.01152}                & {\ul 0.14999}              & 10564.60496               \\ \cmidrule(l){2-12} 
                          & \multirow{2}{*}{Ours}                       & \multicolumn{1}{c|}{$\rm DPR^{-}$}   & \multicolumn{1}{l}{0.018642} & \multicolumn{1}{l}{0.060286} & \multicolumn{1}{l|}{194.172414} & 0.02880             & 0.01981           & \multicolumn{1}{l|}{\textbf{764.64297}} & 0.01275             & 0.14256          & 9352.88490                \\
                          &                                             & \multicolumn{1}{c|}{DPR}   & \multicolumn{1}{l}{\textbf{0.018725}} & \multicolumn{1}{l}{\textbf{0.060328}} & \multicolumn{1}{l|}{\textbf{183.25149}} & \textbf{0.02931}             & \textbf{0.01995}           & \multicolumn{1}{l|}{779.15264} & \textbf{0.01284}             & \textbf{0.14334}          & 9345.51632                \\ \bottomrule
\end{tabular}}
\end{table}

\begin{table}[tp]
\renewcommand\arraystretch{1.25}
    \caption{Overall performance of applying our framework on three common dataset with two recommendation models. The best result in each column is in bold; The best baseline result in each line is underlined. The recommendation performance is evaluated as a ranking task. Higher Recall and NDCG mean better model performance; ARP is a measure of the influence of feedback loops on the model, the lower, the better.}
    \label{tab:5}
\resizebox{\linewidth}{!}{
\begin{tabular}{@{}cccccccccccc@{}}
\toprule
                          &                            &                            & \multicolumn{3}{c}{Movielens 100K}                                                                   & \multicolumn{3}{c}{Lastfm}                                                                            & \multicolumn{3}{c}{Movielens 1M}                                                      \\ \midrule
                          &                            & \multicolumn{1}{c|}{}      & \multicolumn{1}{c}{Recall@5} & \multicolumn{1}{c}{NDCG@5} & \multicolumn{1}{c|}{ARP@5}               & \multicolumn{1}{c}{Recall@5} & \multicolumn{1}{c}{NDCG@5} & \multicolumn{1}{c|}{ARP@5}                & \multicolumn{1}{c}{Recall@5} & \multicolumn{1}{c}{NDCG@5} & \multicolumn{1}{c}{ARP@5} \\ \midrule
\multirow{7}{*}{MF}       & \multirow{5}{*}{Pairwise}  & \multicolumn{1}{c|}{BPR}   & 0.05987                      & 0.07313                    & \multicolumn{1}{l|}{1660.67730}          & 0.04200                      & 0.06496                    & \multicolumn{1}{l|}{11976.18810}          & 0.03276                      & 0.08673                    & 3674.39914                \\
                          &                            & \multicolumn{1}{c|}{UBPR}  & 0.05857                      & 0.06926                    & \multicolumn{1}{l|}{{\ul 1616.57269}}    & 0.05609                      & 0.05694                    & \multicolumn{1}{l|}{10507.83587}          & 0.03272                      & 0.08740                    & {\ul 3665.27997}          \\
                          &                            & \multicolumn{1}{c|}{EBPR}  & 0.04670                      & 0.05991                    & \multicolumn{1}{l|}{1671.57558}          & 0.04351                      & 0.06931                    & \multicolumn{1}{l|}{11946.89653}          & 0.03333                      & 0.08543                    & 3673.60000                \\
                          &                            & \multicolumn{1}{c|}{PDA}   & 0.05542                      & 0.06912                    & \multicolumn{1}{l|}{1669.68297}          & 0.04110                      & 0.06466                    & \multicolumn{1}{l|}{12148.00000}          & 0.03276                      & 0.08672                    & 3674.39990                \\
                          &                            & \multicolumn{1}{c|}{UPL}   & 0.05723                      & 0.06784                    & \multicolumn{1}{l|}{1670.44657}          & {\ul 0.06393}                & {\ul 0.07097}              & \multicolumn{1}{l|}{{\ul 12083.15468}}    & {\ul 0.03424}                & {\ul 0.08853}              & 3674.32333                \\ \cmidrule(l){2-12} 
                          & \multirow{2}{*}{Pointwise} & \multicolumn{1}{c|}{RelMF} & {\ul 0.06826}                & 0.07714                    & \multicolumn{1}{l|}{1633.62087}          & 0.05395                      & 0.05476                    & \multicolumn{1}{l|}{10165.80245}          & 0.03316                      & 0.08724                    & 3672.15789                \\
                          &                            & \multicolumn{1}{c|}{MFDU}  & 0.06382                      & {\ul 0.07768}              & \multicolumn{1}{l|}{1674.16167}          & 0.04171                      & 0.06502                    & \multicolumn{1}{l|}{12146.58093}          & 0.03431                      & 0.08807                    & 3674.80000                \\ \cmidrule(l){2-12} 
                          & \multirow{2}{*}{Ours}      & \multicolumn{1}{c|}{$DPR^{-}$}   & 0.06941            & 0.07825           & \multicolumn{1}{l|}{1601.21017} & 0.07363             & 0.07407           & \multicolumn{1}{l|}{\textbf{980.11035}}   & 0.03533             & 0.08932           & 3662.80000       \\ 
                          &                            & \multicolumn{1}{c|}{DPR}   & \textbf{0.07057}             & \textbf{0.07883}           & \multicolumn{1}{l|}{\textbf{1599.21017}} & \textbf{0.07421}             & \textbf{0.07487}           & \multicolumn{1}{l|}{983.11035}   & \textbf{0.03594}             & \textbf{0.09017}           & \textbf{3659.7259}       \\ 
                          \midrule
\multirow{8}{*}{LightGCN} & \multirow{5}{*}{Pairwise}  & \multicolumn{1}{c|}{BPR}   & 0.05471                      & 0.06306                    & \multicolumn{1}{l|}{1590.23019}          & 0.16280                      & 0.16939                    & \multicolumn{1}{l|}{11876.99043}          & 0.02577                      & 0.05364                    & 3404.32303                \\
                          &                            & \multicolumn{1}{c|}{UBPR}  & 0.04600                      & 0.05211                    & \multicolumn{1}{l|}{{\ul 1555.11584}}    & 0.17884                      & 0.17780                    & \multicolumn{1}{l|}{11682.94342}          & 0.02582                      & 0.05397                    & 3407.13363                \\
                          &                            & \multicolumn{1}{c|}{EBPR}  & {\ul 0.05875}                & 0.06715                    & \multicolumn{1}{l|}{1618.03265}          & 0.10616                      & 0.13413                    & \multicolumn{1}{l|}{11852.51579}          & 0.01761                      & 0.03615                    & 3199.43516                \\
                          &                            & \multicolumn{1}{c|}{PDA}   & 0.05753                      & {\ul 0.06742}              & \multicolumn{1}{l|}{1611.92944}          & 0.14669                      & 0.15424                    & \multicolumn{1}{l|}{11996.39314}          & {\ul 0.04023}                & {\ul 0.08204}              & 3600.08885                \\
                          &                            & \multicolumn{1}{c|}{UPL}   & 0.04331                      & 0.05042                    & \multicolumn{1}{l|}{1568.27109}          & 0.12003                      & 0.11568                    & \multicolumn{1}{l|}{{\ul 11482.49704}}    & 0.03383                      & 0.07373                    & 3537.03534                \\ \cmidrule(l){2-12} 
                          & \multirow{2}{*}{Pointwise} & \multicolumn{1}{c|}{RelMF} & 0.03492                      & 0.04136                    & \multicolumn{1}{l|}{1525.22344}          & {\ul 0.17973}                & {\ul 0.18246}              & \multicolumn{1}{l|}{11668.15456}          & 0.01350                      & 0.02676                    & {\ul 3015.77742}          \\
                          &                            & \multicolumn{1}{c|}{MFDU}  & 0.04548                      & 0.05341                    & \multicolumn{1}{l|}{1625.08715}          & 0.10001                      & 0.10457                    & \multicolumn{1}{l|}{11941.92372}          & 0.02055                      & 0.04108                    & 3245.67513                \\ \cmidrule(l){2-12} 
                          & \multirow{2}{*}{Ours}                       & \multicolumn{1}{c|}{$DPR^{-}$}   & 0.06137            & 0.06835           & \multicolumn{1}{l|}{1490.20107} & 0.18087             & 0.18774         & \multicolumn{1}{l|}{11277.13549} & 0.04194            & 0.08666           & 3106.81232       \\
                          &                                             & \multicolumn{1}{c|}{DPR}   & \textbf{0.06213}            & \textbf{0.06901}           & \multicolumn{1}{l|}{\textbf{1485.33579}} & \textbf{0.18521}             & \textbf{0.18795}          & \multicolumn{1}{l|}{\textbf{11263.15942}} & \textbf{0.04259}             & \textbf{0.08713}           & \textbf{3099.25963}       \\ \bottomrule
\end{tabular}}
\end{table}

\subsection{Overall performance comparison} In Table \ref{tab:4} and \ref{tab:5}, we present the overall performance of applying baseline methods and our method over six datasets on two recommendation backbones. The baseline models include five pairwise ranking methods and two pointwise methods. For all datasets, the ranking performance of our DPR method outperforms the alternative baseline methods.\par
\begin{itemize}[leftmargin = *]
    \item The performance of the IPS-based methods is not stable and even worse than BPR on some datasets. In a major part, because this type of approach requires relatively accurate inference prior, Ref-MF, for example, is an IPS-based method, and a better estimate of the property score might make it closer to the theoretical maximum. EPR adds additional user-item information to enhance explainability, but introducing additional information does not fully exploit the model and thus performs poorly on three common datasets. The IPS-based methods UPR and Rel-MF show worse performance in KuaiRec, as the IPS-based methods cannot well mitigate the effects of exposure mechanisms in lengthy feedback loops.\par
    \item The baseline model fails to achieve the best performance in both debiasing and ranking performance. Baselines forcing the model to include de-biasing as an additional requirement in the optimization options come at the cost of model performance. Instead of adding a forcing requirement in the training process, DPR increases the adaptability of the model and shows the expected performance of the model as much as possible. Thus DPR has shown promising results on different backbones. Although the de-biasing effect of DPR on the Kuairec dataset is poor compared to RelMF, the accuracy of DPR is significantly higher than that of RelMF. The goal of DPR is to achieve the best possible debiasing performance while maintaining accuracy.\par
    \item The proposed plug-in UFN further improves the performance of DPR. UFN reduces spurious information from false negative samples, which helps the model to avoid falling into sub-optimal traps. Further, DPR equipped with UFN can achieve better performance on most datasets, it is possible that unpopular items have a higher probability of being false negative samples, and UFN can help DPR to be more sensitive to unpopular items, thus achieving lower ARP scores.
\end{itemize}
In summary, DPR achieves better performance than other baseline methods in mitigating the negative effects of the feedback loops and unknown exposure mechanisms without sacrificing recommendation performance in most cases.\par
\begin{table}[tp]
\renewcommand\arraystretch{1.25}
		\centering
		\caption{Comparison between with and without the plugin UFN. BPR$^{^+}$ is the BPR assembled with UFN and DPR$^-$ is the DPR without UFN.\;\%Improv.: percentage of improvement on Recall of BPR$^{^+}$, DPR$^-$, DPR over BPR. }
		\label{tab:6}
  \begin{tabular}{cccccc}
\toprule
\multicolumn{2}{c}{Methods}                 & BPR    & BPR$^{^+}$   & DPR$^-$    & DPR     \\
\midrule
\specialrule{0em}{1pt}{1pt}
\multirow{3}{*}{Coat}           & Recall@20 & 0.2408 & 0.2474 & 0.2558  & \textbf{0.2595}  \\
                                & ARP@20    & 27.21  & 26.37  & 25.75   & \textbf{25.55}   \\
                                & \textbf{\%Improv.} & -      & 2.74\% & 6.23\%  & 7.77\%  \\
\midrule
\multirow{3}{*}{Movielens 100K} & Recall@20 & 0.7125 & 0.7233 & 0.7861  & \textbf{0.797}   \\
                                & ARP@20    & 138.21 & 115.38 & 99.57   & \textbf{95.22}   \\
                                & \textbf{\%Improv.} & -      & 1.52\% & 10.33\% & 11.92\%  \\ \bottomrule
\end{tabular}
\end{table}
\subsection{Effectiveness of the proposed plugin UFN}
We conduct experiments on two datasets to show the effectiveness of the proposed plugin for model performance improvement. As shown in Table \ref{tab:6}, deploying UFN on BPR alone can make a certain improvement in the model performance and decreases the weight of popular items in the final recommended items. This proves that UFN can mitigate the problem of false negative samples and improve the performance of the model. Moreover, the DPR also shows better performance than the baseline methods without UFN. This demonstrates that the two modules proposed in this paper are useful for model performance improvement and that the effects can be superimposed. UFN can further improve the performance of the loss function on both BPR and DPR, indicating that it can be widely used in the pairwise loss function to mitigate the interference caused by false negative samples.
\begin{figure}[tp]
		\centering
		\includegraphics[width=0.55\linewidth]{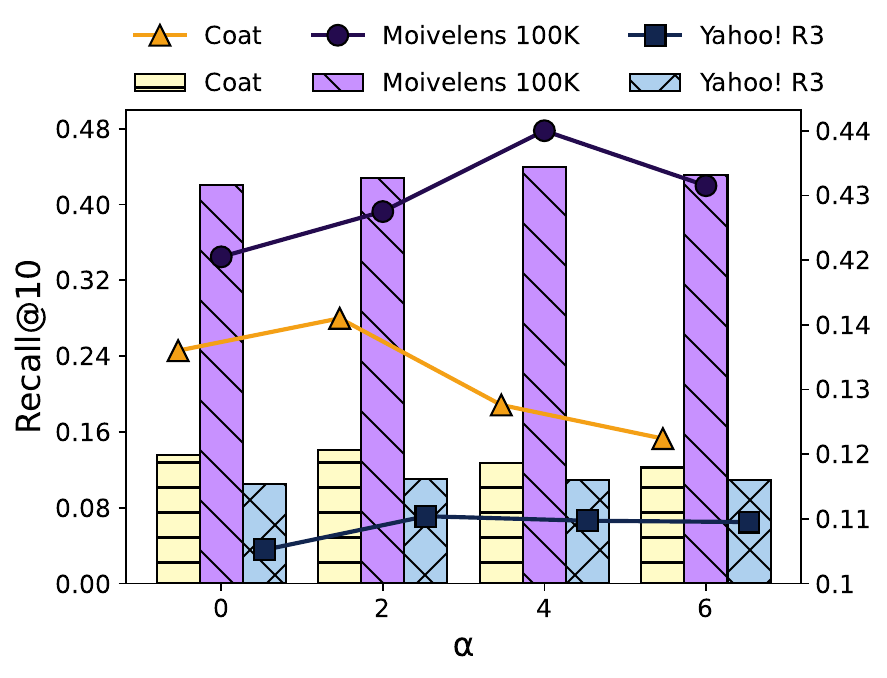}
		\caption{Performance comparison of different levels of $\alpha$. Both y-axes are Recall@10, with the right side used to highlight trends.}
		\label{fig:6}
\end{figure}
\subsection{Variants of DPR loss.} 
We compare the performance of DPR with different levels of feedback loop parameter $\alpha$ on different datasets. As shown in Figure \ref{fig:6}, DPR performs best with $\alpha = 2$ on Coat, which is a dataset with short feedback loops. In addition, DPR performs best with $\alpha = 4$ on Movielens 100K, which is a dataset that suffers from more feedback loops. An increase in the starting $\alpha$ improves the performance of the model on Yahoo! R3, but the model is not as sensitive to an increase in $\alpha$ as it is on Movielens 100K. We speculate that this may be due to the higher sparsity of the data. The different levels of optimal $\alpha$ in the three datasets demonstrate the ability of $\alpha$ to mitigate the different lengths of feedback loops suffered by datasets.\par
Moreover, the performance of DPR deteriorates when the level of alpha continues to increase. The results match our previous analysis of exposure mechanisms, where having the same exposure for all items is not necessarily beneficial for personalized recommendations. Compared to a fully fair exposure mechanism, a personalized user-oriented exposure mechanism is the ideal exposure mechanism. This means that full debiasing does not achieve the best recommendation performance.\par

\begin{table}[tp]
\renewcommand\arraystretch{1.25}
		\centering
		\caption{Comparison between BPR and DPR for different backbones on Movielens 100K; \%Improv.: percentage of improvement on Recall of DPR over BPR. }
		\label{tab:7}
\scalebox{1.0}{
\begin{tabular}{@{}cccccc@{}}
\toprule
\multirow{2}{*}{Methods} & \multicolumn{2}{l}{BPR} & \multicolumn{2}{l}{DPR} & \multirow{2}{*}{\textbf{\%Improv.}} \\ \cmidrule(lr){2-5} 
                         & Recall@20    & ARP@20   & Recall@20    & ARP@20   &     \\ \cmidrule(r){1-5}
MF                       & 0.7125       & 97.02    & 0.7974       & 95.22    & 11.92\%  \\
LightGCN                 & 0.7910       & 108.41    & 0.8176       & 88.76    & 3.36\%   \\
DIN                      & 0.7688       & 125.32    & 0.8057       & 96.42    & 5.07\%    \\
Wide\&Deep               & 0.8021       & 113.44    & 0.8245       & 85.71    & 2.79\%  \\ \bottomrule
\end{tabular}}
\end{table}
\subsection{Diversity backbone experiments} To demonstrate the generalizability of DPR, we conduct comparison experiments on multiple types of backbone. Here, \textbf{MF} is representative of DNN-based models, \textbf{LightGCN} is representative of graph-based models, \textbf{DIN}\cite{ref29} is representative of self-attention-based models, and \textbf{Wide\&Deep}\cite{ref30} is representative of MLP-based models. As shown in Table \ref{tab:7}, DPR demonstrated superior performance with an average improvement of 5.79 \% over BPR on four different backbones. The final results using DPR on MF can compete with and even outperform complex algorithms such as LightGCN using BPR, and complex models such as DIN can be further improved by using DPR, thus demonstrating that DPR can further develop the potential of the model itself and achieve better performance.\par
The outstanding performance of DPR on different classes of baselines highlights that DPR is an effective algorithm that can be widely used in mainstream recommendation models today. We believe that DPR has the potential to replace BPR as a commonly used loss function because of its ability to achieve high accuracy and powerful debiasing effects.\par
\subsection{Universal applicability of UFN and comparison of different levels of \texorpdfstring{$\beta$}.} We deploy UFN on multiple baseline methods to demonstrate its universality. Figure \ref{fig:5} shows that the baseline methods all show a certain improvement after the deployment of UFN. Among them, pairwise-based methods have a more obvious improvement, while the improvement of the pointwise method is smaller. Probably because they have different objectives, pointwise methods place more importance on the accuracy of prediction scores than pairwise methods, and the deployment of UFN may affect the accuracy of model prediction scores. For pairwise-based methods, UFN can help to widen the score gap between positive and negative items, leading to better performance.\par
We compare the performance of DPR and baseline methods with different levels of parameter $\beta$. As shown in Figure \ref{fig:7}, Initially, the model performs better with increasing $\beta$ since a larger beta means that the model focuses more on false negative samples, which improves the model performance. However, the high attention to false negative samples means that the effect of some strongly negative samples is equally compromised so that the model becomes worse when $\beta$ grows beyond a certain point and continues to increase. To our knowledge, items as false negative examples are more likely to be cold items, gaining low exposure due to their low initial popularity and deteriorating as the feedback loop proceeds, transforming into false negative examples because they have no opportunity to interact with users. Thus, the average popularity of the final recommendation keeps decreasing as the beta increases.
\begin{figure}[tp]
\centering
\centering
\centerline{\includegraphics[width = 0.7\linewidth]{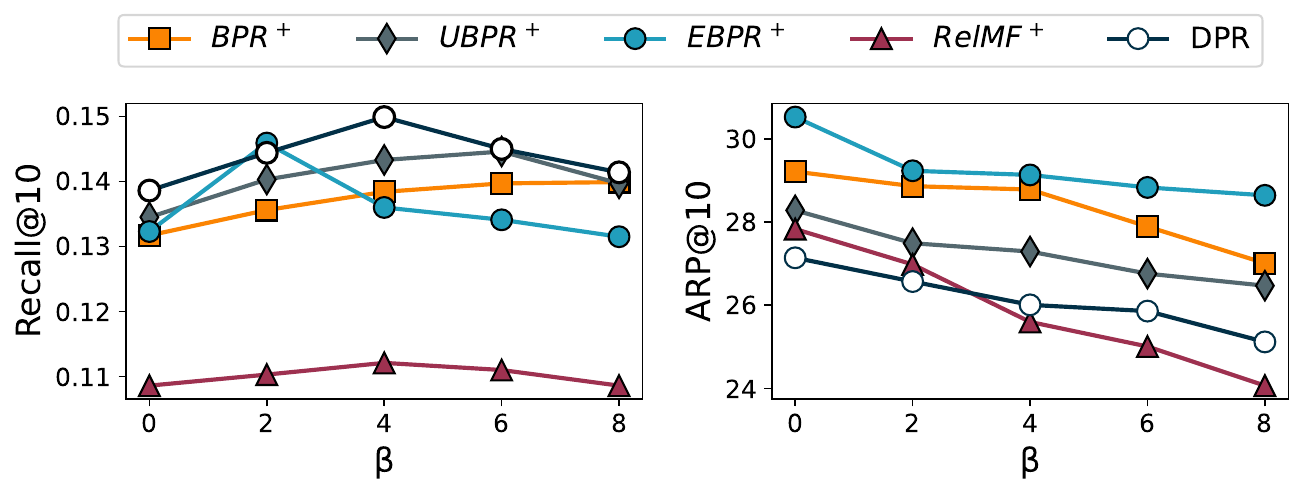}}
\caption{Universal applicability of UFN and comparison of different levels of $\beta$.}%
\label{fig:7}
\end{figure}
\hfill
\begin{figure}[tp]
\centering
\includegraphics[width = 0.7\linewidth]{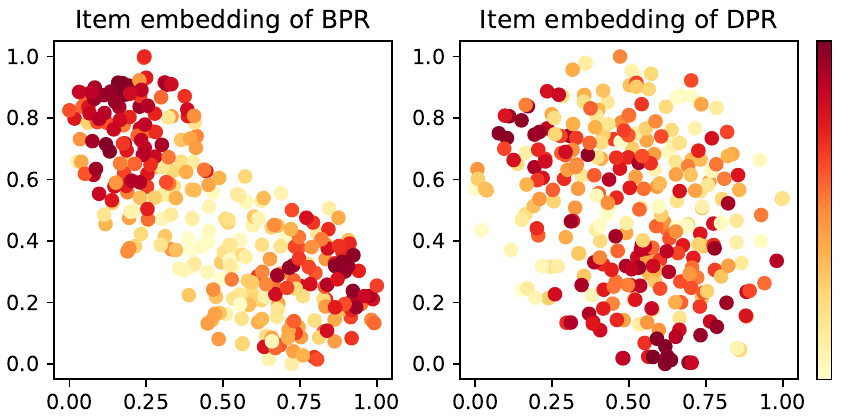}
\caption{Visualization of the learned item embeddings of BPR and DPR.}%
\label{fig:8}
\centering
\end{figure}
\subsection{Interpretability based on model embeddings}
As previously introduced, DPR can help the model mitigate accumulated bias in feedback loops and model the true preferences of users. We investigate whether DPR has such advantages by visualizing the learned item embeddings using t-SNE \cite{ref28}. \par
Figure \ref{fig:8} shows the learned item embeddings of BPR and DPR on MF, where the items with different popularity are painted in different colors. We observe that the embeddings of BPR are significantly separated according to item popularity, with popular items distributed at both ends, meaning that popular items receive more attention from the model. Items with higher popularity are considered high-quality items by the model and thus receive higher exposure. At this point, the dominant factor in the recommendation of the model is not the properties of the item itself, but other external factors such as popularity. The accumulated bias in the feedback data is inherited by the trained model and will be amplified via feedback loops.\par
On the other hand, there is no significant difference between popular and unpopular items in the embeddings of DPR, which means that explicit properties in the dataset, such as item popularity, do not play an influential role in the modeling process, but the true identity of items does. DPR can help reduce the influence of irrelevant factors in the training process and establish more reliable user and item embedding, so as to achieve unbiased recommendations. Visualization of the learned item embeddings illustrates the advantages of DPR and makes reasonable interpretations. 
\section{Conclusion and Future Work}
In this work, we show that exposure mechanisms and feedback loops have a cross-effect on the recommendation model and find the key factor for the transformation of the exposure mechanism under feedback loops. Moreover, we revisit the commonly used pairwise ranking methods and point out that they are biased in modeling the true preferences of users. As a result, we propose DPR, an unbiased pairwise ranking method where the cross-effects of exposure mechanisms and feedback loops are mitigated by simply setting an optimal $\alpha$ without knowing the exact exposure mechanism and the feedback loop length. We conduct experiments on widely used models such as MLPs, and DIN to demonstrate the ease of use and effectiveness of DPR. DPR achieves better performance than other baseline methods in mitigating the negative effects of the feedback loops and unknown exposure mechanisms without sacrificing recommendation performance in most cases. \par
To mitigate the effect of false negative samples, we propose UFN, a removable plug-in with generality. UFN can further improve the performance of the loss function on pairwise-based methods, indicating that it can be widely used in the pairwise loss function to mitigate the interference caused by false negative samples. However, because the pointwise-based methods do not achieve positive versus negative comparison, it does not work well with UFN and therefore fail to achieve the same performance as the pairwise-based methods.\par
However, DPR is not able to exploit the favorable side of the exposure mechanism to improve the performance of personalized recommendations. In the future, we will conduct a deeper analysis of the exposure mechanism to explore and exploit its role in the recommendation system and further explore the identity of false negative samples and propose a more general approach. Moreover, personalized hyperparameters for datasets are demanding for the practical effectiveness of the proposed method, and further research on hyperparameter learnability will be conducted in the future to improve the ease of use of the proposed method.
\begin{acks}
This work was supported by the National Natural Science Foundation of Authority Youth Fund under Grant 62002123, the Key Project of Science and technology development Plan of Jilin Province under Grant 20210201082GX, the Science and Technology project of Education Department of Jilin Province under Grand JJKH20221010KJ, and Science and Technology Department project of Jilin Province under Grant 20230101067JC.
\end{acks}

\bibliographystyle{ACM-Reference-Format}
\bibliography{sample-base}


\begin{thebibliography}{49}


\ifx \showCODEN    \undefined \def \showCODEN     #1{\unskip}     \fi
\ifx \showDOI      \undefined \def \showDOI       #1{#1}\fi
\ifx \showISBNx    \undefined \def \showISBNx     #1{\unskip}     \fi
\ifx \showISBNxiii \undefined \def \showISBNxiii  #1{\unskip}     \fi
\ifx \showISSN     \undefined \def \showISSN      #1{\unskip}     \fi
\ifx \showLCCN     \undefined \def \showLCCN      #1{\unskip}     \fi
\ifx \shownote     \undefined \def \shownote      #1{#1}          \fi
\ifx \showarticletitle \undefined \def \showarticletitle #1{#1}   \fi
\ifx \showURL      \undefined \def \showURL       {\relax}        \fi
\providecommand\bibfield[2]{#2}
\providecommand\bibinfo[2]{#2}
\providecommand\natexlab[1]{#1}
\providecommand\showeprint[2][]{arXiv:#2}

\bibitem[Abdollahpouri et~al\mbox{.}(2019)]%
        {ref11}
\bibfield{author}{\bibinfo{person}{Himan Abdollahpouri}, \bibinfo{person}{Robin
  Burke}, {and} \bibinfo{person}{Bamshad Mobasher}.}
  \bibinfo{year}{2019}\natexlab{}.
\newblock \bibinfo{title}{Managing Popularity Bias in Recommender Systems with
  Personalized Re-Ranking}.
\newblock , \bibinfo{numpages}{413--418}~pages.
\newblock
\urldef\tempurl%
\url{https://aaai.org/ocs/index.php/FLAIRS/FLAIRS19/paper/view/18199}
\showURL{%
\tempurl}


\bibitem[Bonner and Vasile(2018)]%
        {ref13}
\bibfield{author}{\bibinfo{person}{Stephen Bonner} {and}
  \bibinfo{person}{Flavian Vasile}.} \bibinfo{year}{2018}\natexlab{}.
\newblock \showarticletitle{Causal Embeddings for Recommendation}. In
  \bibinfo{booktitle}{\emph{Proceedings of the 12th ACM Conference on
  Recommender Systems}} (Vancouver, British Columbia, Canada)
  \emph{(\bibinfo{series}{RecSys '18})}. \bibinfo{publisher}{Association for
  Computing Machinery}, \bibinfo{address}{New York, NY, USA},
  \bibinfo{pages}{104–112}.
\newblock
\showISBNx{9781450359016}
\urldef\tempurl%
\url{https://doi.org/10.1145/3240323.3240360}
\showDOI{\tempurl}


\bibitem[Ca\~{n}amares and Castells(2018)]%
        {ref18}
\bibfield{author}{\bibinfo{person}{Roc\'{\i}o Ca\~{n}amares} {and}
  \bibinfo{person}{Pablo Castells}.} \bibinfo{year}{2018}\natexlab{}.
\newblock \showarticletitle{Should I Follow the Crowd? A Probabilistic Analysis
  of the Effectiveness of Popularity in Recommender Systems}. In
  \bibinfo{booktitle}{\emph{The 41st International ACM SIGIR Conference on
  Research amp; Development in Information Retrieval}} (Ann Arbor, MI, USA)
  \emph{(\bibinfo{series}{SIGIR '18})}. \bibinfo{publisher}{Association for
  Computing Machinery}, \bibinfo{address}{New York, NY, USA},
  \bibinfo{pages}{415–424}.
\newblock
\showISBNx{9781450356572}
\urldef\tempurl%
\url{https://doi.org/10.1145/3209978.3210014}
\showDOI{\tempurl}


\bibitem[Cai et~al\mbox{.}(2023)]%
        {tois7}
\bibfield{author}{\bibinfo{person}{Desheng Cai}, \bibinfo{person}{Shengsheng
  Qian}, \bibinfo{person}{Quan Fang}, \bibinfo{person}{Jun Hu}, {and}
  \bibinfo{person}{Changsheng Xu}.} \bibinfo{year}{2023}\natexlab{}.
\newblock \showarticletitle{User Cold-Start Recommendation via Inductive
  Heterogeneous Graph Neural Network}.
\newblock \bibinfo{journal}{\emph{ACM Trans. Inf. Syst.}} \bibinfo{volume}{41},
  \bibinfo{number}{3}, Article \bibinfo{articleno}{64} (\bibinfo{date}{feb}
  \bibinfo{year}{2023}), \bibinfo{numpages}{27}~pages.
\newblock
\showISSN{1046-8188}
\urldef\tempurl%
\url{https://doi.org/10.1145/3560487}
\showDOI{\tempurl}


\bibitem[Cantador et~al\mbox{.}(2011)]%
        {lastfm}
\bibfield{author}{\bibinfo{person}{Iv\'{a}n Cantador}, \bibinfo{person}{Peter
  Brusilovsky}, {and} \bibinfo{person}{Tsvi Kuflik}.}
  \bibinfo{year}{2011}\natexlab{}.
\newblock \showarticletitle{2nd Workshop on Information Heterogeneity and
  Fusion in Recommender Systems (HetRec 2011)}. In
  \bibinfo{booktitle}{\emph{Proceedings of the 5th ACM conference on
  Recommender systems}} (Chicago, IL, USA) \emph{(\bibinfo{series}{RecSys
  2011})}. \bibinfo{publisher}{ACM}, \bibinfo{address}{New York, NY, USA}.
\newblock


\bibitem[Chaney et~al\mbox{.}(2018)]%
        {ref16}
\bibfield{author}{\bibinfo{person}{Allison J.~B. Chaney},
  \bibinfo{person}{Brandon~M. Stewart}, {and} \bibinfo{person}{Barbara~E.
  Engelhardt}.} \bibinfo{year}{2018}\natexlab{}.
\newblock \bibinfo{title}{How algorithmic confounding in recommendation systems
  increases homogeneity and decreases utility}.
\newblock , \bibinfo{numpages}{224--232}~pages.
\newblock
\urldef\tempurl%
\url{https://doi.org/10.1145/3240323.3240370}
\showDOI{\tempurl}


\bibitem[Chen et~al\mbox{.}(2023)]%
        {tois6}
\bibfield{author}{\bibinfo{person}{Jiawei Chen}, \bibinfo{person}{Hande Dong},
  \bibinfo{person}{Xiang Wang}, \bibinfo{person}{Fuli Feng},
  \bibinfo{person}{Meng Wang}, {and} \bibinfo{person}{Xiangnan He}.}
  \bibinfo{year}{2023}\natexlab{}.
\newblock \showarticletitle{Bias and Debias in Recommender System: A Survey and
  Future Directions}.
\newblock \bibinfo{journal}{\emph{ACM Trans. Inf. Syst.}} \bibinfo{volume}{41},
  \bibinfo{number}{3}, Article \bibinfo{articleno}{67} (\bibinfo{date}{feb}
  \bibinfo{year}{2023}), \bibinfo{numpages}{39}~pages.
\newblock
\showISSN{1046-8188}
\urldef\tempurl%
\url{https://doi.org/10.1145/3564284}
\showDOI{\tempurl}


\bibitem[Chen et~al\mbox{.}(2021)]%
        {tois5}
\bibfield{author}{\bibinfo{person}{Jiawei Chen}, \bibinfo{person}{Chengquan
  Jiang}, \bibinfo{person}{Can Wang}, \bibinfo{person}{Sheng Zhou},
  \bibinfo{person}{Yan Feng}, \bibinfo{person}{Chun Chen},
  \bibinfo{person}{Martin Ester}, {and} \bibinfo{person}{Xiangnan He}.}
  \bibinfo{year}{2021}\natexlab{}.
\newblock \showarticletitle{CoSam: An Efficient Collaborative Adaptive Sampler
  for Recommendation}.
\newblock \bibinfo{journal}{\emph{ACM Trans. Inf. Syst.}} \bibinfo{volume}{39},
  \bibinfo{number}{3}, Article \bibinfo{articleno}{34} (\bibinfo{date}{may}
  \bibinfo{year}{2021}), \bibinfo{numpages}{24}~pages.
\newblock
\showISSN{1046-8188}
\urldef\tempurl%
\url{https://doi.org/10.1145/3450289}
\showDOI{\tempurl}


\bibitem[Cheng et~al\mbox{.}(2016)]%
        {ref30}
\bibfield{author}{\bibinfo{person}{Heng-Tze Cheng}, \bibinfo{person}{Levent
  Koc}, \bibinfo{person}{Jeremiah Harmsen}, \bibinfo{person}{Tal Shaked},
  \bibinfo{person}{Tushar Chandra}, \bibinfo{person}{Hrishi Aradhye},
  \bibinfo{person}{Glen Anderson}, \bibinfo{person}{Greg Corrado},
  \bibinfo{person}{Wei Chai}, \bibinfo{person}{Mustafa Ispir},
  \bibinfo{person}{Rohan Anil}, \bibinfo{person}{Zakaria Haque},
  \bibinfo{person}{Lichan Hong}, \bibinfo{person}{Vihan Jain},
  \bibinfo{person}{Xiaobing Liu}, {and} \bibinfo{person}{Hemal Shah}.}
  \bibinfo{year}{2016}\natexlab{}.
\newblock \showarticletitle{Wide\&amp; Deep Learning for Recommender Systems}.
  In \bibinfo{booktitle}{\emph{Proceedings of the 1st Workshop on Deep Learning
  for Recommender Systems}} (Boston, MA, USA) \emph{(\bibinfo{series}{DLRS
  2016})}. \bibinfo{publisher}{Association for Computing Machinery},
  \bibinfo{address}{New York, NY, USA}, \bibinfo{pages}{7–10}.
\newblock
\showISBNx{9781450347952}
\urldef\tempurl%
\url{https://doi.org/10.1145/2988450.2988454}
\showDOI{\tempurl}


\bibitem[Damak et~al\mbox{.}(2021)]%
        {ref12}
\bibfield{author}{\bibinfo{person}{Khalil Damak}, \bibinfo{person}{Sami
  Khenissi}, {and} \bibinfo{person}{Olfa Nasraoui}.}
  \bibinfo{year}{2021}\natexlab{}.
\newblock \showarticletitle{Debiased Explainable Pairwise Ranking from Implicit
  Feedback}. In \bibinfo{booktitle}{\emph{Fifteenth ACM Conference on
  Recommender Systems}} (Amsterdam, Netherlands) \emph{(\bibinfo{series}{RecSys
  '21})}. \bibinfo{publisher}{Association for Computing Machinery},
  \bibinfo{address}{New York, NY, USA}, \bibinfo{pages}{321–331}.
\newblock
\showISBNx{9781450384582}
\urldef\tempurl%
\url{https://doi.org/10.1145/3460231.3474274}
\showDOI{\tempurl}


\bibitem[Ding et~al\mbox{.}(2020a)]%
        {ref9}
\bibfield{author}{\bibinfo{person}{Jingtao Ding}, \bibinfo{person}{Yuhan Quan},
  \bibinfo{person}{Quanming Yao}, \bibinfo{person}{Yong Li}, {and}
  \bibinfo{person}{Depeng Jin}.} \bibinfo{year}{2020}\natexlab{a}.
\newblock \bibinfo{title}{Simplify and Robustify Negative Sampling for Implicit
  Collaborative Filtering}.
\newblock
\newblock
\urldef\tempurl%
\url{https://proceedings.neurips.cc/paper/2020/hash/0c7119e3a6a2209da6a5b90e5b5b75bd-Abstract.html}
\showURL{%
\tempurl}


\bibitem[Ding et~al\mbox{.}(2020b)]%
        {tois11}
\bibfield{author}{\bibinfo{person}{Jingtao Ding}, \bibinfo{person}{Guanghui
  Yu}, \bibinfo{person}{Yong Li}, \bibinfo{person}{Xiangnan He}, {and}
  \bibinfo{person}{Depeng Jin}.} \bibinfo{year}{2020}\natexlab{b}.
\newblock \showarticletitle{Improving Implicit Recommender Systems with
  Auxiliary Data}.
\newblock \bibinfo{journal}{\emph{ACM Trans. Inf. Syst.}} \bibinfo{volume}{38},
  \bibinfo{number}{1}, Article \bibinfo{articleno}{11} (\bibinfo{date}{feb}
  \bibinfo{year}{2020}), \bibinfo{numpages}{27}~pages.
\newblock
\showISSN{1046-8188}
\urldef\tempurl%
\url{https://doi.org/10.1145/3372338}
\showDOI{\tempurl}


\bibitem[Fei et~al\mbox{.}(2022)]%
        {tois10}
\bibfield{author}{\bibinfo{person}{Hao Fei}, \bibinfo{person}{Tat-Seng Chua},
  \bibinfo{person}{Chenliang Li}, \bibinfo{person}{Donghong Ji},
  \bibinfo{person}{Meishan Zhang}, {and} \bibinfo{person}{Yafeng Ren}.}
  \bibinfo{year}{2022}\natexlab{}.
\newblock \showarticletitle{On the Robustness of Aspect-Based Sentiment
  Analysis: Rethinking Model, Data, and Training}.
\newblock \bibinfo{journal}{\emph{ACM Trans. Inf. Syst.}} \bibinfo{volume}{41},
  \bibinfo{number}{2}, Article \bibinfo{articleno}{50} (\bibinfo{date}{dec}
  \bibinfo{year}{2022}), \bibinfo{numpages}{32}~pages.
\newblock
\showISSN{1046-8188}
\urldef\tempurl%
\url{https://doi.org/10.1145/3564281}
\showDOI{\tempurl}


\bibitem[Gao et~al\mbox{.}(2022)]%
        {kuai}
\bibfield{author}{\bibinfo{person}{Chongming Gao}, \bibinfo{person}{Shijun Li},
  \bibinfo{person}{Wenqiang Lei}, \bibinfo{person}{Jiawei Chen},
  \bibinfo{person}{Biao Li}, \bibinfo{person}{Peng Jiang},
  \bibinfo{person}{Xiangnan He}, \bibinfo{person}{Jiaxin Mao}, {and}
  \bibinfo{person}{Tat-Seng Chua}.} \bibinfo{year}{2022}\natexlab{}.
\newblock \showarticletitle{KuaiRec: A Fully-Observed Dataset and Insights for
  Evaluating Recommender Systems} \emph{(\bibinfo{series}{CIKM '22})}.
  \bibinfo{publisher}{Association for Computing Machinery},
  \bibinfo{address}{New York, NY, USA}, \bibinfo{pages}{540–550}.
\newblock
\showISBNx{9781450392365}
\urldef\tempurl%
\url{https://doi.org/10.1145/3511808.3557220}
\showDOI{\tempurl}


\bibitem[Gao et~al\mbox{.}(2023)]%
        {tois1}
\bibfield{author}{\bibinfo{person}{Chongming Gao}, \bibinfo{person}{Shiqi
  Wang}, \bibinfo{person}{Shijun Li}, \bibinfo{person}{Jiawei Chen},
  \bibinfo{person}{Xiangnan He}, \bibinfo{person}{Wenqiang Lei},
  \bibinfo{person}{Biao Li}, \bibinfo{person}{Yuan Zhang}, {and}
  \bibinfo{person}{Peng Jiang}.} \bibinfo{year}{2023}\natexlab{}.
\newblock \showarticletitle{CIRS: Bursting Filter Bubbles by Counterfactual
  Interactive Recommender System}.
\newblock \bibinfo{journal}{\emph{ACM Trans. Inf. Syst.}} (\bibinfo{date}{apr}
  \bibinfo{year}{2023}).
\newblock
\showISSN{1046-8188}
\urldef\tempurl%
\url{https://doi.org/10.1145/3594871}
\showDOI{\tempurl}


\bibitem[Harper and Konstan(2016)]%
        {ml-100k}
\bibfield{author}{\bibinfo{person}{F.~Maxwell Harper} {and}
  \bibinfo{person}{Joseph~A. Konstan}.} \bibinfo{year}{2016}\natexlab{}.
\newblock \showarticletitle{The MovieLens Datasets: History and Context}.
\newblock \bibinfo{journal}{\emph{{ACM} Trans. Interact. Intell. Syst.}}
  \bibinfo{volume}{5}, \bibinfo{number}{4} (\bibinfo{year}{2016}),
  \bibinfo{pages}{19:1--19:19}.
\newblock
\urldef\tempurl%
\url{https://doi.org/10.1145/2827872}
\showDOI{\tempurl}


\bibitem[He et~al\mbox{.}(2020)]%
        {ref27}
\bibfield{author}{\bibinfo{person}{Xiangnan He}, \bibinfo{person}{Kuan Deng},
  \bibinfo{person}{Xiang Wang}, \bibinfo{person}{Yan Li},
  \bibinfo{person}{YongDong Zhang}, {and} \bibinfo{person}{Meng Wang}.}
  \bibinfo{year}{2020}\natexlab{}.
\newblock \showarticletitle{LightGCN: Simplifying and Powering Graph
  Convolution Network for Recommendation}. In
  \bibinfo{booktitle}{\emph{Proceedings of the 43rd International ACM SIGIR
  Conference on Research and Development in Information Retrieval}} (Virtual
  Event, China) \emph{(\bibinfo{series}{SIGIR '20})}.
  \bibinfo{publisher}{Association for Computing Machinery},
  \bibinfo{address}{New York, NY, USA}, \bibinfo{pages}{639–648}.
\newblock
\showISBNx{9781450380164}
\urldef\tempurl%
\url{https://doi.org/10.1145/3397271.3401063}
\showDOI{\tempurl}


\bibitem[Hern{\'{a}}ndez{-}Lobato et~al\mbox{.}(2014)]%
        {ref24}
\bibfield{author}{\bibinfo{person}{Jos{\'{e}}~Miguel Hern{\'{a}}ndez{-}Lobato},
  \bibinfo{person}{Neil Houlsby}, {and} \bibinfo{person}{Zoubin Ghahramani}.}
  \bibinfo{year}{2014}\natexlab{}.
\newblock \showarticletitle{Probabilistic Matrix Factorization with Non-random
  Missing Data}. In \bibinfo{booktitle}{\emph{Proceedings of the 31th
  International Conference on Machine Learning, {ICML} 2014, Beijing, China,
  21-26 June 2014}} \emph{(\bibinfo{series}{{JMLR} Workshop and Conference
  Proceedings}, Vol.~\bibinfo{volume}{32})}. \bibinfo{publisher}{JMLR.org},
  \bibinfo{pages}{1512--1520}.
\newblock
\urldef\tempurl%
\url{http://proceedings.mlr.press/v32/hernandez-lobatob14.html}
\showURL{%
\tempurl}


\bibitem[Imran et~al\mbox{.}(2023)]%
        {tois8}
\bibfield{author}{\bibinfo{person}{Mubashir Imran}, \bibinfo{person}{Hongzhi
  Yin}, \bibinfo{person}{Tong Chen}, \bibinfo{person}{Quoc Viet~Hung Nguyen},
  \bibinfo{person}{Alexander Zhou}, {and} \bibinfo{person}{Kai Zheng}.}
  \bibinfo{year}{2023}\natexlab{}.
\newblock \showarticletitle{ReFRS: Resource-Efficient Federated Recommender
  System for Dynamic and Diversified User Preferences}.
\newblock \bibinfo{journal}{\emph{ACM Trans. Inf. Syst.}} \bibinfo{volume}{41},
  \bibinfo{number}{3}, Article \bibinfo{articleno}{65} (\bibinfo{date}{feb}
  \bibinfo{year}{2023}), \bibinfo{numpages}{30}~pages.
\newblock
\showISSN{1046-8188}
\urldef\tempurl%
\url{https://doi.org/10.1145/3560486}
\showDOI{\tempurl}


\bibitem[Kalimeris et~al\mbox{.}(2021)]%
        {ref19}
\bibfield{author}{\bibinfo{person}{Dimitris Kalimeris}, \bibinfo{person}{Smriti
  Bhagat}, \bibinfo{person}{Shankar Kalyanaraman}, {and} \bibinfo{person}{Udi
  Weinsberg}.} \bibinfo{year}{2021}\natexlab{}.
\newblock \showarticletitle{Preference Amplification in Recommender Systems}.
  In \bibinfo{booktitle}{\emph{Proceedings of the 27th ACM SIGKDD Conference on
  Knowledge Discovery amp; Data Mining}} (Virtual Event, Singapore)
  \emph{(\bibinfo{series}{KDD '21})}. \bibinfo{publisher}{Association for
  Computing Machinery}, \bibinfo{address}{New York, NY, USA},
  \bibinfo{pages}{805–815}.
\newblock
\showISBNx{9781450383325}
\urldef\tempurl%
\url{https://doi.org/10.1145/3447548.3467298}
\showDOI{\tempurl}


\bibitem[Koren et~al\mbox{.}(2009)]%
        {ref26}
\bibfield{author}{\bibinfo{person}{Yehuda Koren}, \bibinfo{person}{Robert
  Bell}, {and} \bibinfo{person}{Chris Volinsky}.}
  \bibinfo{year}{2009}\natexlab{}.
\newblock \showarticletitle{Matrix Factorization Techniques for Recommender
  Systems}.
\newblock \bibinfo{journal}{\emph{Computer}} \bibinfo{volume}{42},
  \bibinfo{number}{8} (\bibinfo{year}{2009}), \bibinfo{pages}{30--37}.
\newblock
\urldef\tempurl%
\url{https://doi.org/10.1109/MC.2009.263}
\showDOI{\tempurl}


\bibitem[Krauth et~al\mbox{.}(2020)]%
        {ref2}
\bibfield{author}{\bibinfo{person}{Karl Krauth}, \bibinfo{person}{Sarah Dean},
  \bibinfo{person}{Alex Zhao}, \bibinfo{person}{Wenshuo Guo},
  \bibinfo{person}{Mihaela Curmei}, \bibinfo{person}{Benjamin Recht}, {and}
  \bibinfo{person}{Michael~I. Jordan}.} \bibinfo{year}{2020}\natexlab{}.
\newblock \bibinfo{title}{Do Offline Metrics Predict Online Performance in
  Recommender Systems?}
\newblock
\newblock
\showeprint{arXiv:2011.07931}


\bibitem[Krauth et~al\mbox{.}(2022)]%
        {ref3}
\bibfield{author}{\bibinfo{person}{Karl Krauth}, \bibinfo{person}{Yixin Wang},
  {and} \bibinfo{person}{Michael~I. Jordan}.} \bibinfo{year}{2022}\natexlab{}.
\newblock \bibinfo{title}{Breaking Feedback Loops in Recommender Systems with
  Causal Inference}.
\newblock
\newblock
\showeprint{arXiv:2207.01616}


\bibitem[Lee et~al\mbox{.}(2021)]%
        {mfdu}
\bibfield{author}{\bibinfo{person}{Jae-woong Lee}, \bibinfo{person}{Seongmin
  Park}, {and} \bibinfo{person}{Jongwuk Lee}.} \bibinfo{year}{2021}\natexlab{}.
\newblock \showarticletitle{Dual Unbiased Recommender Learning for Implicit
  Feedback} \emph{(\bibinfo{series}{SIGIR '21})}.
  \bibinfo{publisher}{Association for Computing Machinery},
  \bibinfo{address}{New York, NY, USA}, \bibinfo{pages}{1647–1651}.
\newblock
\showISBNx{9781450380379}
\urldef\tempurl%
\url{https://doi.org/10.1145/3404835.3463118}
\showDOI{\tempurl}


\bibitem[Liang et~al\mbox{.}(2016)]%
        {ref21}
\bibfield{author}{\bibinfo{person}{Dawen Liang}, \bibinfo{person}{Laurent
  Charlin}, \bibinfo{person}{James McInerney}, {and} \bibinfo{person}{David~M.
  Blei}.} \bibinfo{year}{2016}\natexlab{}.
\newblock \showarticletitle{Modeling User Exposure in Recommendation}. In
  \bibinfo{booktitle}{\emph{Proceedings of the 25th International Conference on
  World Wide Web}} (Montr\'{e}al, Qu\'{e}bec, Canada)
  \emph{(\bibinfo{series}{WWW '16})}. \bibinfo{publisher}{International World
  Wide Web Conferences Steering Committee}, \bibinfo{address}{Republic and
  Canton of Geneva, CHE}, \bibinfo{pages}{951–961}.
\newblock
\showISBNx{9781450341431}
\urldef\tempurl%
\url{https://doi.org/10.1145/2872427.2883090}
\showDOI{\tempurl}


\bibitem[Liu et~al\mbox{.}(2023)]%
        {tois3}
\bibfield{author}{\bibinfo{person}{Dugang Liu}, \bibinfo{person}{Pengxiang
  Cheng}, \bibinfo{person}{Zinan Lin}, \bibinfo{person}{Xiaolian Zhang},
  \bibinfo{person}{Zhenhua Dong}, \bibinfo{person}{Rui Zhang},
  \bibinfo{person}{Xiuqiang He}, \bibinfo{person}{Weike Pan}, {and}
  \bibinfo{person}{Zhong Ming}.} \bibinfo{year}{2023}\natexlab{}.
\newblock \showarticletitle{Bounding System-Induced Biases in Recommender
  Systems with a Randomized Dataset}.
\newblock \bibinfo{journal}{\emph{ACM Trans. Inf. Syst.}} \bibinfo{volume}{41},
  \bibinfo{number}{4}, Article \bibinfo{articleno}{108} (\bibinfo{date}{apr}
  \bibinfo{year}{2023}), \bibinfo{numpages}{26}~pages.
\newblock
\showISSN{1046-8188}
\urldef\tempurl%
\url{https://doi.org/10.1145/3582002}
\showDOI{\tempurl}


\bibitem[Mansoury et~al\mbox{.}(2020)]%
        {ref17}
\bibfield{author}{\bibinfo{person}{Masoud Mansoury}, \bibinfo{person}{Himan
  Abdollahpouri}, \bibinfo{person}{Mykola Pechenizkiy},
  \bibinfo{person}{Bamshad Mobasher}, {and} \bibinfo{person}{Robin Burke}.}
  \bibinfo{year}{2020}\natexlab{}.
\newblock \showarticletitle{Feedback Loop and Bias Amplification in Recommender
  Systems}. In \bibinfo{booktitle}{\emph{Proceedings of the 29th ACM
  International Conference on Information amp; Knowledge Management}} (Virtual
  Event, Ireland) \emph{(\bibinfo{series}{CIKM '20})}.
  \bibinfo{publisher}{Association for Computing Machinery},
  \bibinfo{address}{New York, NY, USA}, \bibinfo{pages}{2145–2148}.
\newblock
\showISBNx{9781450368599}
\urldef\tempurl%
\url{https://doi.org/10.1145/3340531.3412152}
\showDOI{\tempurl}


\bibitem[Marlin and Zemel(2009)]%
        {ref22}
\bibfield{author}{\bibinfo{person}{Benjamin~M. Marlin} {and}
  \bibinfo{person}{Richard~S. Zemel}.} \bibinfo{year}{2009}\natexlab{}.
\newblock \showarticletitle{Collaborative prediction and ranking with
  non-random missing data}. In \bibinfo{booktitle}{\emph{Proceedings of the
  2009 {ACM} Conference on Recommender Systems, RecSys 2009, New York, NY, USA,
  October 23-25, 2009}}, \bibfield{editor}{\bibinfo{person}{Lawrence~D.
  Bergman}, \bibinfo{person}{Alexander Tuzhilin}, \bibinfo{person}{Robin~D.
  Burke}, \bibinfo{person}{Alexander Felfernig}, {and} \bibinfo{person}{Lars
  Schmidt{-}Thieme}} (Eds.). \bibinfo{publisher}{{ACM}},
  \bibinfo{pages}{5--12}.
\newblock
\urldef\tempurl%
\url{https://doi.org/10.1145/1639714.1639717}
\showDOI{\tempurl}


\bibitem[Ren et~al\mbox{.}(2023)]%
        {upl}
\bibfield{author}{\bibinfo{person}{Yi Ren}, \bibinfo{person}{Hongyan Tang},
  \bibinfo{person}{Jiangpeng Rong}, {and} \bibinfo{person}{Siwen Zhu}.}
  \bibinfo{year}{2023}\natexlab{}.
\newblock \showarticletitle{Unbiased Pairwise Learning from Implicit Feedback
  for Recommender Systems without Biased Variance Control}.
\newblock \bibinfo{howpublished}{Proceedings of the 46th International ACM
  SIGIR Conference on Research and Development in Information Retrieval, 2023}.
\newblock  (\bibinfo{year}{2023}).
\newblock
\urldef\tempurl%
\url{https://doi.org/10.1145/3539618.3592077}
\showDOI{\tempurl}
\showeprint{arXiv:2304.05066}


\bibitem[Rendle et~al\mbox{.}(2009)]%
        {ref8}
\bibfield{author}{\bibinfo{person}{Steffen Rendle}, \bibinfo{person}{Christoph
  Freudenthaler}, \bibinfo{person}{Zeno Gantner}, {and} \bibinfo{person}{Lars
  Schmidt{-}Thieme}.} \bibinfo{year}{2009}\natexlab{}.
\newblock \showarticletitle{{BPR:} Bayesian Personalized Ranking from Implicit
  Feedback}. In \bibinfo{booktitle}{\emph{{UAI} 2009, Proceedings of the
  Twenty-Fifth Conference on Uncertainty in Artificial Intelligence, Montreal,
  QC, Canada, June 18-21, 2009}}, \bibfield{editor}{\bibinfo{person}{Jeff~A.
  Bilmes} {and} \bibinfo{person}{Andrew~Y. Ng}} (Eds.).
  \bibinfo{publisher}{{AUAI} Press}, \bibinfo{pages}{452--461}.
\newblock
\urldef\tempurl%
\url{https://dslpitt.org/uai/displayArticleDetails.jsp?mmnu=1\&smnu=2\&article\_id=1630\&proceeding\_id=25}
\showURL{%
\tempurl}


\bibitem[Saito(2020)]%
        {ref4}
\bibfield{author}{\bibinfo{person}{Yuta Saito}.}
  \bibinfo{year}{2020}\natexlab{}.
\newblock \showarticletitle{Unbiased Pairwise Learning from Biased Implicit
  Feedback} \emph{(\bibinfo{series}{ICTIR '20})}.
  \bibinfo{publisher}{Association for Computing Machinery},
  \bibinfo{address}{New York, NY, USA}, \bibinfo{pages}{5–12}.
\newblock
\showISBNx{9781450380676}
\urldef\tempurl%
\url{https://doi.org/10.1145/3409256.3409812}
\showDOI{\tempurl}


\bibitem[Saito et~al\mbox{.}(2020a)]%
        {ref5}
\bibfield{author}{\bibinfo{person}{Yuta Saito}, \bibinfo{person}{Suguru
  Yaginuma}, \bibinfo{person}{Yuta Nishino}, \bibinfo{person}{Hayato Sakata},
  {and} \bibinfo{person}{Kazuhide Nakata}.} \bibinfo{year}{2020}\natexlab{a}.
\newblock \showarticletitle{Unbiased Recommender Learning from
  Missing-Not-At-Random Implicit Feedback}. In
  \bibinfo{booktitle}{\emph{Proceedings of the 13th International Conference on
  Web Search and Data Mining}} (Houston, TX, USA) \emph{(\bibinfo{series}{WSDM
  '20})}. \bibinfo{publisher}{Association for Computing Machinery},
  \bibinfo{address}{New York, NY, USA}, \bibinfo{pages}{501–509}.
\newblock
\showISBNx{9781450368223}
\urldef\tempurl%
\url{https://doi.org/10.1145/3336191.3371783}
\showDOI{\tempurl}


\bibitem[Saito et~al\mbox{.}(2020b)]%
        {ref15}
\bibfield{author}{\bibinfo{person}{Yuta Saito}, \bibinfo{person}{Suguru
  Yaginuma}, \bibinfo{person}{Yuta Nishino}, \bibinfo{person}{Hayato Sakata},
  {and} \bibinfo{person}{Kazuhide Nakata}.} \bibinfo{year}{2020}\natexlab{b}.
\newblock \showarticletitle{Unbiased Recommender Learning from
  Missing-Not-At-Random Implicit Feedback}. In \bibinfo{booktitle}{\emph{{WSDM}
  '20: The Thirteenth {ACM} International Conference on Web Search and Data
  Mining, Houston, TX, USA, February 3-7, 2020}},
  \bibfield{editor}{\bibinfo{person}{James Caverlee},
  \bibinfo{person}{Xia~(Ben) Hu}, \bibinfo{person}{Mounia Lalmas}, {and}
  \bibinfo{person}{Wei Wang}} (Eds.). \bibinfo{publisher}{{ACM}},
  \bibinfo{pages}{501--509}.
\newblock
\urldef\tempurl%
\url{https://doi.org/10.1145/3336191.3371783}
\showDOI{\tempurl}


\bibitem[Schmit and Riquelme(2018)]%
        {ref1}
\bibfield{author}{\bibinfo{person}{Sven Schmit} {and} \bibinfo{person}{Carlos
  Riquelme}.} \bibinfo{year}{2018}\natexlab{}.
\newblock \showarticletitle{Human Interaction with Recommendation Systems}. In
  \bibinfo{booktitle}{\emph{Proceedings of the Twenty-First International
  Conference on Artificial Intelligence and Statistics}}
  \emph{(\bibinfo{series}{Proceedings of Machine Learning Research},
  Vol.~\bibinfo{volume}{84})}, \bibfield{editor}{\bibinfo{person}{Amos Storkey}
  {and} \bibinfo{person}{Fernando Perez-Cruz}} (Eds.).
  \bibinfo{publisher}{PMLR}, \bibinfo{pages}{862--870}.
\newblock


\bibitem[Schnabel et~al\mbox{.}(2016)]%
        {ref6}
\bibfield{author}{\bibinfo{person}{Tobias Schnabel}, \bibinfo{person}{Adith
  Swaminathan}, \bibinfo{person}{Ashudeep Singh}, \bibinfo{person}{Navin
  Chandak}, {and} \bibinfo{person}{Thorsten Joachims}.}
  \bibinfo{year}{2016}\natexlab{}.
\newblock \showarticletitle{Recommendations as Treatments: Debiasing Learning
  and Evaluation}. In \bibinfo{booktitle}{\emph{Proceedings of the 33nd
  International Conference on Machine Learning, {ICML} 2016, New York City, NY,
  USA, June 19-24, 2016}} \emph{(\bibinfo{series}{{JMLR} Workshop and
  Conference Proceedings}, Vol.~\bibinfo{volume}{48})},
  \bibfield{editor}{\bibinfo{person}{Maria{-}Florina Balcan} {and}
  \bibinfo{person}{Kilian~Q. Weinberger}} (Eds.).
  \bibinfo{publisher}{JMLR.org}, \bibinfo{pages}{1670--1679}.
\newblock
\urldef\tempurl%
\url{http://proceedings.mlr.press/v48/schnabel16.html}
\showURL{%
\tempurl}


\bibitem[Sinha et~al\mbox{.}(2017)]%
        {ref23}
\bibfield{author}{\bibinfo{person}{Ayan Sinha}, \bibinfo{person}{David~F.
  Gleich}, {and} \bibinfo{person}{Karthik Ramani}.}
  \bibinfo{year}{2017}\natexlab{}.
\newblock \showarticletitle{Deconvolving Feedback Loops in Recommender
  Systems}.
\newblock \bibinfo{journal}{\emph{CoRR}}  \bibinfo{volume}{abs/1703.01049}
  (\bibinfo{year}{2017}).
\newblock
\showeprint[arXiv]{1703.01049}
\urldef\tempurl%
\url{http://arxiv.org/abs/1703.01049}
\showURL{%
\tempurl}


\bibitem[Sun et~al\mbox{.}(2019)]%
        {ref20}
\bibfield{author}{\bibinfo{person}{Wenlong Sun}, \bibinfo{person}{Sami
  Khenissi}, \bibinfo{person}{Olfa Nasraoui}, {and} \bibinfo{person}{Patrick
  Shafto}.} \bibinfo{year}{2019}\natexlab{}.
\newblock \showarticletitle{Debiasing the Human-Recommender System Feedback
  Loop in Collaborative Filtering}. In \bibinfo{booktitle}{\emph{Companion
  Proceedings of The 2019 World Wide Web Conference}} (San Francisco, USA)
  \emph{(\bibinfo{series}{WWW '19})}. \bibinfo{publisher}{Association for
  Computing Machinery}, \bibinfo{address}{New York, NY, USA},
  \bibinfo{pages}{645–651}.
\newblock
\showISBNx{9781450366755}
\urldef\tempurl%
\url{https://doi.org/10.1145/3308560.3317303}
\showDOI{\tempurl}


\bibitem[van~der Maaten and Hinton(2008)]%
        {ref28}
\bibfield{author}{\bibinfo{person}{Laurens van~der Maaten} {and}
  \bibinfo{person}{Geoffrey Hinton}.} \bibinfo{year}{2008}\natexlab{}.
\newblock \showarticletitle{Visualizing Data using t-SNE}.
\newblock \bibinfo{journal}{\emph{Journal of Machine Learning Research}}
  \bibinfo{volume}{9}, \bibinfo{number}{86} (\bibinfo{year}{2008}),
  \bibinfo{pages}{2579--2605}.
\newblock
\urldef\tempurl%
\url{http://jmlr.org/papers/v9/vandermaaten08a.html}
\showURL{%
\tempurl}


\bibitem[Wan et~al\mbox{.}(2022)]%
        {ref7}
\bibfield{author}{\bibinfo{person}{Qi Wan}, \bibinfo{person}{Xiangnan He},
  \bibinfo{person}{Xiang Wang}, \bibinfo{person}{Jiancan Wu},
  \bibinfo{person}{Wei Guo}, {and} \bibinfo{person}{Ruiming Tang}.}
  \bibinfo{year}{2022}\natexlab{}.
\newblock \showarticletitle{Cross Pairwise Ranking for Unbiased Item
  Recommendation}. In \bibinfo{booktitle}{\emph{{WWW} '22: The {ACM} Web
  Conference 2022, Virtual Event, Lyon, France, April 25 - 29, 2022}},
  \bibfield{editor}{\bibinfo{person}{Fr{\'{e}}d{\'{e}}rique Laforest},
  \bibinfo{person}{Rapha{\"{e}}l Troncy}, \bibinfo{person}{Elena Simperl},
  \bibinfo{person}{Deepak Agarwal}, \bibinfo{person}{Aristides Gionis},
  \bibinfo{person}{Ivan Herman}, {and} \bibinfo{person}{Lionel M{\'{e}}dini}}
  (Eds.). \bibinfo{publisher}{{ACM}}, \bibinfo{pages}{2370--2378}.
\newblock
\urldef\tempurl%
\url{https://doi.org/10.1145/3485447.3512010}
\showDOI{\tempurl}


\bibitem[Wang et~al\mbox{.}(2021)]%
        {ref10}
\bibfield{author}{\bibinfo{person}{Jinpeng Wang}, \bibinfo{person}{Jieming
  Zhu}, {and} \bibinfo{person}{Xiuqiang He}.} \bibinfo{year}{2021}\natexlab{}.
\newblock \showarticletitle{Cross-Batch Negative Sampling for Training
  Two-Tower Recommenders}. In \bibinfo{booktitle}{\emph{{SIGIR} '21: The 44th
  International {ACM} {SIGIR} Conference on Research and Development in
  Information Retrieval, Virtual Event, Canada, July 11-15, 2021}},
  \bibfield{editor}{\bibinfo{person}{Fernando Diaz}, \bibinfo{person}{Chirag
  Shah}, \bibinfo{person}{Torsten Suel}, \bibinfo{person}{Pablo Castells},
  \bibinfo{person}{Rosie Jones}, {and} \bibinfo{person}{Tetsuya Sakai}} (Eds.).
  \bibinfo{publisher}{{ACM}}, \bibinfo{pages}{1632--1636}.
\newblock
\urldef\tempurl%
\url{https://doi.org/10.1145/3404835.3463032}
\showDOI{\tempurl}


\bibitem[Wang et~al\mbox{.}(2023a)]%
        {tois2}
\bibfield{author}{\bibinfo{person}{Wenjie Wang}, \bibinfo{person}{Xinyu Lin},
  \bibinfo{person}{Liuhui Wang}, \bibinfo{person}{Fuli Feng},
  \bibinfo{person}{Yunshan Ma}, {and} \bibinfo{person}{Tat-Seng Chua}.}
  \bibinfo{year}{2023}\natexlab{a}.
\newblock \showarticletitle{Causal Disentangled Recommendation Against User
  Preference Shifts}.
\newblock \bibinfo{journal}{\emph{ACM Trans. Inf. Syst.}} (\bibinfo{date}{apr}
  \bibinfo{year}{2023}).
\newblock
\showISSN{1046-8188}
\urldef\tempurl%
\url{https://doi.org/10.1145/3593022}
\showDOI{\tempurl}
\newblock
\shownote{Just Accepted}.


\bibitem[Wang et~al\mbox{.}(2020)]%
        {ref25}
\bibfield{author}{\bibinfo{person}{Yixin Wang}, \bibinfo{person}{Dawen Liang},
  \bibinfo{person}{Laurent Charlin}, {and} \bibinfo{person}{David~M. Blei}.}
  \bibinfo{year}{2020}\natexlab{}.
\newblock \showarticletitle{Causal Inference for Recommender Systems}. In
  \bibinfo{booktitle}{\emph{RecSys 2020: Fourteenth {ACM} Conference on
  Recommender Systems, Virtual Event, Brazil, September 22-26, 2020}},
  \bibfield{editor}{\bibinfo{person}{Rodrygo L.~T. Santos},
  \bibinfo{person}{Leandro~Balby Marinho}, \bibinfo{person}{Elizabeth~M. Daly},
  \bibinfo{person}{Li~Chen}, \bibinfo{person}{Kim Falk}, \bibinfo{person}{Noam
  Koenigstein}, {and} \bibinfo{person}{Edleno~Silva de~Moura}} (Eds.).
  \bibinfo{publisher}{{ACM}}, \bibinfo{pages}{426--431}.
\newblock
\urldef\tempurl%
\url{https://doi.org/10.1145/3383313.3412225}
\showDOI{\tempurl}


\bibitem[Wang et~al\mbox{.}(2023b)]%
        {tois9}
\bibfield{author}{\bibinfo{person}{Yifan Wang}, \bibinfo{person}{Weizhi Ma},
  \bibinfo{person}{Min Zhang}, \bibinfo{person}{Yiqun Liu}, {and}
  \bibinfo{person}{Shaoping Ma}.} \bibinfo{year}{2023}\natexlab{b}.
\newblock \showarticletitle{A Survey on the Fairness of Recommender Systems}.
\newblock \bibinfo{journal}{\emph{ACM Trans. Inf. Syst.}} \bibinfo{volume}{41},
  \bibinfo{number}{3}, Article \bibinfo{articleno}{52} (\bibinfo{date}{feb}
  \bibinfo{year}{2023}), \bibinfo{numpages}{43}~pages.
\newblock
\showISSN{1046-8188}
\urldef\tempurl%
\url{https://doi.org/10.1145/3547333}
\showDOI{\tempurl}


\bibitem[Xu et~al\mbox{.}(2022a)]%
        {xu01}
\bibfield{author}{\bibinfo{person}{Yuanbo Xu}, \bibinfo{person}{En Wang},
  \bibinfo{person}{Yongjian Yang}, {and} \bibinfo{person}{Yi Chang}.}
  \bibinfo{year}{2022}\natexlab{a}.
\newblock \showarticletitle{A Unified Collaborative Representation Learning for
  Neural-Network Based Recommender Systems}.
\newblock \bibinfo{journal}{\emph{IEEE Transactions on Knowledge and Data
  Engineering}} \bibinfo{volume}{34}, \bibinfo{number}{11}
  (\bibinfo{year}{2022}), \bibinfo{pages}{5126--5139}.
\newblock
\urldef\tempurl%
\url{https://doi.org/10.1109/TKDE.2021.3054782}
\showDOI{\tempurl}


\bibitem[Xu et~al\mbox{.}(2022b)]%
        {xu02}
\bibfield{author}{\bibinfo{person}{Yuanbo Xu}, \bibinfo{person}{Yongjian Yang},
  \bibinfo{person}{En Wang}, \bibinfo{person}{Fuzhen Zhuang}, {and}
  \bibinfo{person}{Hui Xiong}.} \bibinfo{year}{2022}\natexlab{b}.
\newblock \showarticletitle{Detect Professional Malicious User With Metric
  Learning in Recommender Systems}.
\newblock \bibinfo{journal}{\emph{IEEE Transactions on Knowledge and Data
  Engineering}} \bibinfo{volume}{34}, \bibinfo{number}{9}
  (\bibinfo{year}{2022}), \bibinfo{pages}{4133--4146}.
\newblock
\urldef\tempurl%
\url{https://doi.org/10.1109/TKDE.2020.3040618}
\showDOI{\tempurl}


\bibitem[Zhang et~al\mbox{.}(2021)]%
        {pda}
\bibfield{author}{\bibinfo{person}{Yang Zhang}, \bibinfo{person}{Fuli Feng},
  \bibinfo{person}{Xiangnan He}, \bibinfo{person}{Tianxin Wei},
  \bibinfo{person}{Chonggang Song}, \bibinfo{person}{Guohui Ling}, {and}
  \bibinfo{person}{Yongdong Zhang}.} \bibinfo{year}{2021}\natexlab{}.
\newblock \showarticletitle{Causal Intervention for Leveraging Popularity Bias
  in Recommendation}. In \bibinfo{booktitle}{\emph{Proceedings of the 44th
  International ACM SIGIR Conference on Research and Development in Information
  Retrieval}} (Virtual Event, Canada) \emph{(\bibinfo{series}{SIGIR '21})}.
  \bibinfo{publisher}{Association for Computing Machinery},
  \bibinfo{address}{New York, NY, USA}, \bibinfo{pages}{11–20}.
\newblock
\showISBNx{9781450380379}
\urldef\tempurl%
\url{https://doi.org/10.1145/3404835.3462875}
\showDOI{\tempurl}


\bibitem[Zheng et~al\mbox{.}(2023)]%
        {tois4}
\bibfield{author}{\bibinfo{person}{Ruiqi Zheng}, \bibinfo{person}{Liang Qu},
  \bibinfo{person}{Bin Cui}, \bibinfo{person}{Yuhui Shi}, {and}
  \bibinfo{person}{Hongzhi Yin}.} \bibinfo{year}{2023}\natexlab{}.
\newblock \showarticletitle{AutoML for Deep Recommender Systems: A Survey}.
\newblock \bibinfo{journal}{\emph{ACM Trans. Inf. Syst.}} \bibinfo{volume}{41},
  \bibinfo{number}{4}, Article \bibinfo{articleno}{101} (\bibinfo{date}{mar}
  \bibinfo{year}{2023}), \bibinfo{numpages}{38}~pages.
\newblock
\showISSN{1046-8188}
\urldef\tempurl%
\url{https://doi.org/10.1145/3579355}
\showDOI{\tempurl}


\bibitem[Zheng et~al\mbox{.}(2021)]%
        {ref14}
\bibfield{author}{\bibinfo{person}{Yu Zheng}, \bibinfo{person}{Chen Gao},
  \bibinfo{person}{Xiang Li}, \bibinfo{person}{Xiangnan He},
  \bibinfo{person}{Yong Li}, {and} \bibinfo{person}{Depeng Jin}.}
  \bibinfo{year}{2021}\natexlab{}.
\newblock \showarticletitle{Disentangling User Interest and Conformity for
  Recommendation with Causal Embedding}. In
  \bibinfo{booktitle}{\emph{Proceedings of the Web Conference 2021}}
  (Ljubljana, Slovenia) \emph{(\bibinfo{series}{WWW '21})}.
  \bibinfo{publisher}{Association for Computing Machinery},
  \bibinfo{address}{New York, NY, USA}, \bibinfo{pages}{2980–2991}.
\newblock
\showISBNx{9781450383127}
\urldef\tempurl%
\url{https://doi.org/10.1145/3442381.3449788}
\showDOI{\tempurl}


\bibitem[Zhou et~al\mbox{.}(2018)]%
        {ref29}
\bibfield{author}{\bibinfo{person}{Guorui Zhou}, \bibinfo{person}{Xiaoqiang
  Zhu}, \bibinfo{person}{Chenru Song}, \bibinfo{person}{Ying Fan},
  \bibinfo{person}{Han Zhu}, \bibinfo{person}{Xiao Ma},
  \bibinfo{person}{Yanghui Yan}, \bibinfo{person}{Junqi Jin},
  \bibinfo{person}{Han Li}, {and} \bibinfo{person}{Kun Gai}.}
  \bibinfo{year}{2018}\natexlab{}.
\newblock \showarticletitle{Deep Interest Network for Click-Through Rate
  Prediction}. In \bibinfo{booktitle}{\emph{Proceedings of the 24th ACM SIGKDD
  International Conference on Knowledge Discovery \&amp; Data Mining}} (London,
  United Kingdom) \emph{(\bibinfo{series}{KDD '18})}.
  \bibinfo{publisher}{Association for Computing Machinery},
  \bibinfo{address}{New York, NY, USA}, \bibinfo{pages}{1059–1068}.
\newblock
\showISBNx{9781450355520}
\urldef\tempurl%
\url{https://doi.org/10.1145/3219819.3219823}
\showDOI{\tempurl}


\end{thebibliography}

\end{document}